\documentclass[a4paper,UKenglish,cleveref, autoref, thm-restate]{lipics-v2021}

\usepackage{todonotes,complexity}

\hideLIPIcs 

\usepackage{accents}
\newcommand\thickbar[1]{\accentset{\rule{.45em}{.9pt}}{#1}}

\title{Parameterized Complexity of Efficient Sortation}  
\titlerunning{Parameterized Complexity of Efficient Sortation}  
\author{Robert Ganian}{Algorithms and Complexity Group, TU Wien, Austria}{rganian@ac.tuwien.ac.at}{0000-0002-7762-8045}{Austrian Science Fund (FWF) Project 10.55776/Y1329}
\author{Hung P.\ Hoang}{Algorithms and Complexity Group, TU Wien, Austria}{phoang@ac.tuwien.ac.at}{0000-0001-7883-4134}{Austrian Science Fund (FWF) Project 10.55776/ESP1136425}
\author{Simon Wietheger}{Algorithms and Complexity Group, TU Wien, Austria}{swietheger@ac.tuwien.ac.at}{0000-0002-0734-0708}{Austrian Science Fund (FWF) Project 10.55776/Y1329}
 \authorrunning{R.\ Ganian, H.\,P.\ Hoang, S.\ Wietheger}  
 \Copyright{Robert Ganian, Hung P.\ Hoang, Simon Wietheger}  
 
 \acknowledgements{We would like to thank the anonymous reviewers for their helpful comments.} 
 
\nolinenumbers

\usepackage{amsmath,amsfonts,amsthm,amssymb,mathtools}
\usepackage[algoruled, vlined, linesnumbered]{algorithm2e}
 
\usepackage{tabularx, environ}
\usepackage{tcolorbox}
\usepackage{makecell,booktabs,multirow}
\usepackage{tablefootnote}

\usepackage{mdframed}
\usepackage{xpatch}
\usepackage{floatrow}

\newcommand{\bigO}[1]{{\ensuremath{\mathcal{O}\!\left(#1\right)}}\xspace}

\newtheorem{fact}[theorem]{Fact}

\newcommand{\N}{\mathbb{N}}

\newcommand{\I}{{\cal I}}
\newcommand{\T}{{\cal T}}

\newcommand{\calP}{{\cal P}}
\newcommand{\calQ}{{\cal Q}}

\newcommand{\set}[1]{{\{#1\}}}

\newcommand{\ceil}[1]{{\left\lceil #1 \right\rceil}}

\newcommand{\fpt}{\textup{FPT}\xspace}
\newcommand{\xp}{\textup{XP}\xspace}

\newcommand{\np}{\NP\xspace}

\newcommand{\nphard}{\np-hard\xspace}

\newcommand{\nphardness}{\np-hardness\xspace}

\newcommand{\pathl}{\ensuremath{p}\xspace}
\newcommand{\pvarp}{\textsc{\textup{MD-RSPP}\textsubscript{\textsc{\textup{PL}}}}\xspace}
\newcommand{\pvar}{\textsc{\textup{MD-RSPP}}\xspace}
\newcommand{\pfixed}{\textsc{\textup{MD-SPP}}\xspace}

\newcommand{\pfixedgen}{\textsc{\textup{SMD-SPP}}\xspace}
\newcommand{\ram}{\textup{\textsc{Ram}}}

\newcommand{\threePart}{\textsc{3-Partition}\xspace}

\newcommand{\threeSAT}{\textsc{3-SAT}\xspace}
\newcommand{\threeSATtwo}{\textsc{3-SAT-(2,2)}\xspace}

\newcommand{\tc}{{\cal T}}

\newcommand{\undund}[1]{\mathcal{G}(#1)}

\newcommand{\True}{TRUE\xspace}
\newcommand{\False}{FALSE\xspace}

\newcommand{\myproblem}[3]{
\noindent
\fbox{   \begin{minipage}{\dimexpr\linewidth-2\fboxsep-2\fboxrule\relax}

    #1\par
    \vspace{-3mm}

    \begin{list}{}{       \setlength{\labelwidth}{3.5em}       \setlength{\leftmargin}{\labelwidth}       \addtolength{\leftmargin}{\labelsep}       \setlength{\rightmargin}{.8em} 	  \setlength{\itemsep}{0.5em}      }
      \item[\textbf{Input:}] #2
      \item[\textbf{Task:}] #3
    \end{list}

  \end{minipage} }
}

 \definecolor{cb_orange}{rgb}{0.859 0.427 0.0}
\definecolor{cb_blue}{rgb}{0.0 0.427 0.859}
\definecolor{cb_violet}{rgb}{0.714 0.427 1.0}
\definecolor{cb_dark_seal}{rgb}{0.0 0.286 0.286}
\definecolor{cb_seal}{rgb}{0.0 0.573 0.573}
\definecolor{cb_pink}{rgb}{1.0 0.427 0.714}
\definecolor{cb_rose}{rgb}{1.0 0.714 0.859}
\definecolor{cb_purple}{rgb}{0.286 0.0 0.573}
\definecolor{cb_light_blue}{rgb}{0.427 0.714 1.0}
\definecolor{cb_vlight_blue}{rgb}{0.714 0.859 1.0}
\definecolor{cb_red}{rgb}{0.573 0.0 0.0}
\definecolor{cb_brown}{rgb}{0.573 0.286 0.0}
\definecolor{cb_green}{rgb}{0.141 1.0 0.141}
\definecolor{cb_yellow}{rgb}{1.0 1.0 0.427}

\definecolor{cb2_orange}{rgb}{0.961, 0.475, 0.227}
\definecolor{cb2_purple}{rgb}{0.663, 0.353, 0.631}
\definecolor{cb2_cyan}{rgb}{0.522, 0.753, 0.976}
\definecolor{cb2_blue}{rgb}{0.059, 0.125, 0.502}



\begin{document}

\maketitle

 \begin{abstract}
    A crucial challenge arising in the design of large-scale logistical networks is to optimize parcel sortation for routing.
    We study this problem under the recent graph-theoretic formalization of Van Dyk, Klause, Koenemann and Megow (IPCO 2024). The problem asks---given an input digraph $D$ (the \emph{fulfillment network}) together with a set of \emph{commodities} represented as source-sink tuples---for a minimum-outdegree subgraph $H$ of the transitive closure of $D$ that contains a source-sink route for each of the commodities. Given the underlying motivation, we study two variants of the problem which differ in whether the routes for the commodities are fixed or can be chosen arbitrarily.    
    We perform a thorough parameterized analysis of the complexity of both problems, concentrating on three fundamental parameterizations: 
    \begin{enumerate}
    \item When considering the target outdegree of $H$, we show that the problems are para\NP-hard even in highly restricted cases;
    \item When parameterizing by the number of commodities, we utilize Ramsey-type arguments and the color-coding technique to obtain fixed-parameter algorithms for both problems;
    \item When parameterizing by the structure of $D$, we establish fixed-parameter tractability for both problems w.r.t.\ the combined parameterization of treewidth, maximum degree and the maximum routing length. We complement this with lower bounds which  show that omitting any of the three parameters results in para\NP-hardness.
    \end{enumerate}
\end{abstract}

\newpage
 \section{Introduction}\label{sec:introduction}
The task of finding optimal solutions to logistical challenges has motivated the study of a wide range of computational graph problems including, e.g., the classical \textsc{Vertex} and \textsc{Edge Disjoint Paths}~\cite{GanianO21,GanianOR21,GolovachT11,LokshtanovMP0Z20,Lokshtanov0Z20} problems and \textsc{Coordinated Motion Planning} (also known as \textsc{Multiagent Pathfinding})~\cite{DeligkasEGK024,EibenGK23,FioravantesKKMO24,GeftH22,KornhauserMS84,SternSFK0WLA0KB19}. And yet, when dealing with logistical challenges at a higher scale, collision avoidance (the main goal in these three problems) is no longer relevant and one needs to consider different factors when optimizing or designing a logistical network. In this paper, we focus on \emph{parcel sortation}, a central aspect of contemporary large-scale logistical networks which has not yet been thoroughly investigated from an algorithmic and complexity-theoretic perspective.

In the considered setting, we are given an underlying fulfillment network and a set of commodities each represented as a source and destination node. The nodes in the fulfillment network typically represent facilities at various locations, and each commodity needs to be routed from its current facility $s_i$ (e.g., a large warehouse) to a facility $t_i$ in the vicinity of the end customer. However, when routing parceled commodities in the network, each parcel that travels from $s_i$ to $t_i$ via some internal node $u$ must be ``sorted'' at $u$ based on its subsequent downstream node. Hence, if multiple commodities arrive at $u$ and each needs to be routed to a different facility downstream, it is necessary to apply \emph{sortation} at $u$ to subdivide the stream of incoming parcels between the next stops; in large-scale operations $u$ would typically be equipped with a designated \emph{sort point} for each downstream node, and the number of sort points that a node can feasibly have is typically limited. On the other hand, if all the commodities arriving at $u$ were to be routed to the same downstream node, one can avoid the costly sortation step at $u$ (via applying \emph{containerization} at a previous facility and using a process called \emph{cross-docking} at $u$). We refer readers interested in a more detailed description of these processes to previous works on the topic~\cite{Belle2012CrossdockingSO,ChenCHC21,LaraKNS23}. 
  
In their recent article, Van Dyk, Klause, Koenemann and Megow~\cite{VanDyk_2023} have shown that the task of optimizing parcel sortation in a logistical network can be modeled as a surprisingly ``clean'' digraph problem. 
Indeed, if one represents the network as a digraph $D$ and each commodity as $(s_i, t_i, P_i)$ where $P_i$ is an $s_i$-$t_i$-path in $D$, the aim is to find a \emph{sorting network}\footnote{These logistical networks should not be confused with number sorting networks~\cite{AjtaiKS83}.}---a subgraph of the transitive closure of $D$---with minimum outdegree. 
The sorting network captures the information of sort points at each node: in particular, it specifies the downstream nodes that the node has a sort point for.
As we want to control the number of sort points at each node, this translates to the objective of minimizing the outdegree. 
Van Dyk, Klause, Koenemann and Megow~\cite{VanDyk_2023} primarily focused their work on the task of computing the sorting plans when the physical routes of the commodities are fixed. In particular, they established \nphardness even when the underlying undirected graph is a star and obtained several polynomial time algorithms for restricted classes of instances, including an exact algorithm if all commodities start at the same vertex and approximation algorithms for various cases where the underlying network is a tree.

In this work we consider both the lower-level optimization task with fixed physical routes as well as the higher-level optimization task where we can determine the routes and sort points. 
This gives rise to the following two problem formulations, where $\T(D)$ denotes the transitive closure of a directed graph $D$:

\vspace{2mm}
\myproblem{\textsc{Min-Degree Sort Point Problem} (\pfixed)}
{A digraph $D=(V,E)$, a target $T\in \N$ and a set $K$ of routed commodities each of which is a tuple of the form $(s,t,P)$ where $s, t \in V$ and $P$ is an $s$-$t$-path in~$D$.}
{Is there a subgraph $H$ of the transitive closure $\tc(D)$ of $D$ such that the maximum outdegree of $H$ is at most $T$, and for each commodity $(s,t,P) \in K$, there is a directed $s$-$t$-path in $H \cap \tc(P)$?}
\vspace{2mm}

\myproblem{\textsc{Min-Degree Routing and Sort Point Problem} (\pvar)}
{A digraph $D=(V,E)$, a target $T\in \N$ and a set $K\subseteq V\times V$ of commodities.}
{Is there a subgraph $H$ of $\tc(D)$ such that the maximum outdegree of $H$ is at most $T$ and for each commodity $(s,t) \in K$, there is a directed $s$-$t$-path in~$H$?}
\vspace{0.2cm}

Figure~\ref{fig:example_network} provides an example illustrating both problems.\footnote{For complexity-theoretic reasons, here we consider the decision variants. However, all obtained algorithms can solve the corresponding optimization tasks as well.
}
   For \pvar\ the paths of commodities are not fixed, so commodities are defined as 2-tuples. 
 
\begin{figure}[bt]  	\centering
	\includegraphics[page=1, scale=1]{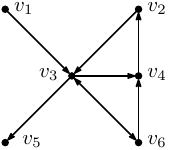}
	\hfill
	\includegraphics[page=2, scale=1]{example_network.pdf}
	\hfill
	\includegraphics[page=3, scale=1]{example_network.pdf}
	\hfill
	\includegraphics[page=4, scale=1]{example_network.pdf}
	\caption{Example network for \pfixed and \pvar. 
	Multi-edges are only for illustration; a solution graph $H$ is in fact simple.	\\
	The left-most image shows an input graph $D$. \\
	The second image visualizes five commodities for \pfixed: $(v_1,v_2, \allowbreak (v_1,v_3,v_6,v_4,v_2)), (v_2,v_5,(v_2,v_3,v_5)), (v_3,v_2,(v_3,v_4,v_2)), (v_3,v_4,(v_3,v_4)),\allowbreak (v_3,v_6,(v_3,v_6))$.\\
	The third image gives a minimum target $(T=2)$ solution graph $H$ to this \pfixed instance.\\
	The right-most image illustrates that for \pvar, by not fixing the paths of commodities, solutions with a smaller target might exist (here $T=1$).	}
	\label{fig:example_network}
	\end{figure}
While both \pfixed\ and \pvar\ could be rather easily shown to be \NP-complete on general graphs, Van Dyk, Klause, Koenemann and Megow showed that the former problem remains \NP-complete even when restricted to orientations of stars~\cite{VanDyk_2023}. In the rest of their article, they then focused on obtaining approximation as well as exact algorithms for \pfixed\ on special classes of oriented trees. Apart from these individual results, the computational complexity of \pvar\ and \pfixed remains entirely unexplored.

  \subsection{Contributions}
The central question tackled by our work is whether the structural properties of inputs can be leveraged to circumvent the \NP-hardness of \pvar\ and \pfixed. To do so, we analyze the problems through the more refined lens of \emph{parameterized complexity}~\cite{CyganFKLMPPS15,DowneyFellows13}. There, the general aim is to design algorithms with running times which are not exponential in the whole input size, but only exponential in some well-defined integer \emph{parameters} of the input. The most desirable complexity class in the parameterized setting is \emph{fixed-parameter tractable} (\FPT), which means that the problem can be solved in time $f(p)\cdot n^{\mathcal{O}(1)}$ where $f(p)$ is a computable function of the parameter $p$ and $n$ is the input size. Naturally, if we establish \NP-hardness for a problem even under the restriction that the parameter $p$ is fixed to a constant, then this immediately rules out fixed-parameter tractability w.r.t.\ $p$; 
 in this case, we say that the problem is para\NP-hard.

 In our work, we consider three natural types of parameterizations for \pvar\ and \pfixed: parameterizing by the target solution quality (i.e., $T$), by the number of commodities, or by the structural properties of the input graph. The article is accordingly split into three parts, one for each of these perspectives.
 
In the first part of our article---Section~\ref{sec:target}---we establish that \pfixed and \pvar remain \NP-hard already when $T=2$. We complement this lower bound by providing a polynomial-time algorithm solving the latter problem in the special case of $T=1$, meaning that for \pvar\ the \NP-hardness result is tight. While the proofs of these results are non-trivial and especially the algorithm requires us to obtain deeper insights into the problem, they do not rely on advanced technical machinery and can thus be seen as introductory. We note that the provided hardness result holds already for orientations of trees with a single source, and complements an earlier additive $(+1)$-approximation of Van Dyk, Klause, Koenemann and Megow for \pfixed\ on this graph class~\cite[Theorem 3]{VanDyk_2023}.

In Section~\ref{sec:commodities}, we turn towards a parameterized analysis of the considered problems w.r.t.\ the number of commodities. 
For \pfixed, we begin by establishing fixed-parameter tractability for the case of $T=1$ via a stand-alone algorithm, as our more involved arguments for the general cases do not seamlessly transfer to this simpler setting.
 We then obtain an involved fixed-parameter algorithm for the ``more general'' setting of $T\geq 2$ by a kernelization subroutine that is based on Ramsey-type arguments, but with a twist: to implement the main reduction rule used here, we need to transition to a more general variant of the problem where some of the vertices do not contribute towards the degree bound $T$.
    
As the algorithm used for \pfixed relies on the fixed paths of the commodities, separate arguments have to be used in order to solve \pvar\ when $T\geq 2$.
Here, we obtain a fixed-parameter algorithm as well, this time by combining a delicate branching routine for solution templates with a color coding approach.  
 It is perhaps worth noting that---in spite of some superficial similarity between the problems---the fixed-parameter tractability of our problems contrasts the known para\NP-hardness of \textsc{Vertex Disjoint Paths} (parameterized by the number of paths) on directed graphs~\cite{FortuneHW80,LopesS22}.

In Section~\ref{sec:structural}, we target graph-structural parameterizations which would yield tractability for arbitrary choices of $T$ and arbitrarily many commodities. Given the previously established \NP-hardness of \pfixed on orientations of stars---and the fact that the same reduction also works for \pvar---one could ask whether parameterizing by treewidth plus the maximum vertex degree (of the underlying undirected graph) suffices; after all, this combined parameterization has already been successfully employed to achieve fixed-parameter algorithms for a number of other challenging problems, including, e.g., \textsc{Edge Disjoint Paths}~\cite{GanianOR21}. Unfortunately, our reductions in Section~\ref{sec:target} already establish the \NP-hardness of both problems of interest even on bounded-degree trees.

As our final contribution, we show that the intractability of both problems can be overcome if one takes the maximum length of any admissible route as a third parameter. In particular, we devise a fixed-parameter algorithm that relies on dynamic programming to solve both \pfixed\ and \pvar\ when parameterized by the treewidth and maximum degree of the input graph (or, more precisely, its underlying undirected graph) plus the maximum length of a route in a solution. We complement this algorithm with reductions which prove that all three of these parameters are necessary: dropping any of the three leads to para\NP-hardness for both problems. 

A summary of our complexity results for the two problems is provided in \cref{table:results}. A careful reader may notice that although the techniques used to obtain our algorithms and lower bounds often vary between the two problems, their complexity under the comprehensive range of parameterizations considered here seems to be the same. We believe this outcome to be rather unexpected, and provide a further discussion in Section~\ref{sec:conc}.

\begin{table}[t]
        		\centering
		
 		\begin{tabular}{cccr}
		\toprule
		 		Parameter & Variant & Complexity & Reference\\
		\hline
		Target $T$ & \pvar & Polynomial time solvable for $T \leq 1$ & Thm.~\ref{thm:fixed_target1} \\
		& both & \nphard for $T\geq 2$ & Thm.~\ref{thm:pvar_pfix_deg_tw} \\
		\hline
		Number $|K|$ 
		 		& \multirow{2}{*}{both} & \multirow{2}{*}{\fpt} & Thm.~\ref{thm:var_commodities_target1}, \ref{thm:fixed_commodities}, \\
		of Commodities  & & & \&~\ref{thm:pvar_commodities_fpt} \\
		\hline
		\multirow{4}{*}{Structural}  & \multirow{4}{*}{both} & \fpt w.r.t.\ degree + treewidth + path length & Thm.~\ref{thm:deg_pathlen_tw}\\
    								 		 & & para\nphard w.r.t.\ degree + treewidth & Thm.~\ref{thm:pvar_pfix_deg_tw}\\
    								 		      										 &  & para\nphard w.r.t.\ treewidth + path length & \cite{VanDyk_2023} \& Fact~\ref{fact:pvar_pfix_pathlen_tw}\\
    										 & & para\nphard w.r.t.\ degree + path length & Thm.~\ref{thm:pvar_deg_pathlen} \&~\ref{thm:pfixed_deg_pathlen}\\
		\bottomrule
    \end{tabular}
  \vspace{-0.15cm}
    \caption{Complexity results for \pfixed and \pvar by different parameterizations. The \emph{degree} and \emph{treewidth} refer to the underlying undirected graph of the input graph $D$. The \emph{path length} is the maximum length of any path fulfilling a commodity. \vspace{-0.35cm}}
    \label{table:results}
\end{table}

  \section{Preliminaries}
For $k \in \N$, we denote by $[k]$ the set $\set{1, \dots, k}$. 

\subsection{Graph Terminology}
We employ standard graph-theoretic terminology~\cite{Diestel}.
For a directed graph $D = (V,E)$, let $\Delta_D^+(v)$ and $\Delta_D^-(v)$ be the sets of out-neighbors and in-neighbors of $v$, respectively, for all $v\in V$.
Let $\delta_D^+(v) = |\Delta_D^+(v)|$, $\delta_D^-(v) = |\Delta_D^-(v)|$, $\delta^+(D) = \max_{v\in V} \delta_D^+(v)$, and $\delta^-(D) = \max_{v\in V} \delta_D^-(v)$. 
We may omit the subscript where $D$ is clear from context.  
 For $u,v \in V$, a path in $D$ from $u$ to $v$ is called a \emph{$u$-$v$-path}.
For two graphs $D, D'$, we say that $D'$ is a subgraph of $D$ and write $D'\subseteq D$ if $V(D')\subseteq V(D)$ and $E(D') \subseteq E(D)$.
The induced subgraph $D[X]$ of a vertex set $X\subseteq V(D)$ is defined by $V(D[X]) = X$ and $E(D[X]) = (X\times X) \cap E(D)$.
 The \emph{transitive closure} of $D$, denoted by $\tc(D)$, is the directed graph on the vertex set $V(D)$, whereas an edge $(i,j)$ exists in $\tc(D)$ if and only if there is a directed path from $i$ to $j$ in $D$.
The \emph{underlying undirected graph} of a directed graph $D$, denoted by $\undund{D}$, is the undirected simple graph obtained from $D$ by replacing each directed edge with an undirected one (all directed edges between a pair of vertices $v$ and $w$ become a single edge $\set{v,w}$).
A \emph{strongly connected component} (SCC) of $D$ is a maximal set of vertices such that every vertex in the set is reachable from all other vertices in the set.

\subsection{Parameterized Complexity}
In parameterized complexity~\cite{CyganFKLMPPS15,DowneyFellows13}, the complexity of a problem is studied not only with respect to the input size, but also with respect to some problem parameter(s). 
The core idea behind parameterized complexity is that the combinatorial explosion resulting from the \NP-hardness of a problem can sometimes be confined to certain structural parameters that are small in practical settings. 
Formal definitions are provided below.

A {\it parameterized problem} $Q$ is a subset of $\Omega^* \times \N$, where $\Omega$ is a fixed alphabet. 
Each instance of $Q$ is a pair $(I, \kappa)$, where $\kappa \in \N$ is called the {\it parameter}. 
A parameterized problem $Q$ is {\it fixed-parameter tractable} (\fpt) if there is an algorithm, called a {\em fixed-parameter algorithm},  that decides whether an input $(I, \kappa)$ is a member of $Q$ in time $f(\kappa) \cdot |I|^{\mathcal{O}(1)}$, where $f$ is a computable function and $|I|$ is the input instance size.  
The class \fpt{} denotes the class of all fixed-parameter tractable parameterized problems.

The class \xp{} contains parameterized problems that can be solved in time  $\bigO{|I|^{f(\kappa)}}$, where $f$ is a computable function.
We say that a parameterized problem is para\NP-hard if it remains \NP-hard even when restricted to instances with a fixed value of the parameter.

\subsection{Treewidth}
A \emph{nice tree decomposition} of an undirected graph $G = (V,E)$ is a pair $(\mathtt{T}, \chi)$, where $\mathtt{T}$ is a tree (whose vertices are called \emph{nodes}) rooted at a node $t_r$ and $\chi$ is a function that assigns each node $t$ a set $\chi(t) \subseteq V$, called a \emph{bag}, such that the following hold:
\begin{itemize}
	\item For every ${u,v} \in E$, there is a node $t$ such that $u,v \in \chi(t)$.
	\item The union over all bags is $V$.
	\item For every vertex $v \in V$, the set of nodes $t$ s.t.\ $v \in \chi(t)$ forms a subtree of $\mathtt{T}$.
	\item $|\chi(\ell)| = 0$ for every leaf $\ell$ of $\mathtt{T}$ and $|\chi(t_r)| = 0$.
	\item There are only three kinds of non-leaf nodes in $\mathtt{T}$:
	\begin{itemize}
		\item \textsf{introduce}: a node $t$ with one child $t'$ such that $\chi(t) = \chi(t') \cup \set{v}$ for some $v \notin \chi(t')$.
		\item \textsf{forget}: a node $t$ with exactly one child $t'$ such that $\chi(t) = \chi(t') \setminus \set{v}$ for some $v \in \chi(t')$.
		\item \textsf{join}: a node $t$ with two children $t_1, t_2$ such that $\chi(t) = \chi(t_1) = \chi(t_2)$.
	\end{itemize}
\end{itemize}

We use $c(t)$ to denote the set of all vertices of $G$ which occur in the bag of $t$ or some descendant of $t$. The width of a nice tree decomposition $(\mathtt{T}, \chi)$ is the size of the largest bag $\chi(t)$ minus 1, and the \emph{treewidth} of $G$ is the minimum width of a nice tree decomposition of~$G$.

\subsection{Problem-Specific Definitions}
For a routed commodity $(s,t,P)$ in an \pfixed instance or a commodity $(s,t)$ in an \pvar instance, we call $s$ its \emph{source} and $t$ its \emph{destination}.
    A graph $H$ satisfying all conditions given in the problem statements is called a \emph{solution graph}. 
Given an \pfixed or \pvar instance $(D,K,T)$ and a commodity $\kappa \in K$, a \emph{witness path} of $\kappa$ is a path in $\tc(D)$ that satisfies the condition of the corresponding problem description for that commodity.
In other words, if $(D,K,T)$ is an \pfixed instance, then a witness path of a commodity $(s,t,P)$ is an $s$-$t$-path in $\tc(P)$.
If $(D,K,T)$ is an \pvar instance, then a witness path of a commodity $(s,t)$ is an $s$-$t$-path in $\tc(D)$.
A \emph{solution path cover} of $(D,K,T)$ is a set of witness paths, one for each commodity in $K$, such that its union is a solution graph of $(D,K,T)$.
           We provide an observation on trivial instances of the problems. 
If the target is at least the maximum outdegree of the input graph, then we are guaranteed to have a YES-instance since the input graph can be trivially used as the solution graph. 
The same also holds if the target is at least the maximum number of commodities starting at any vertex, as there is then a solution by only using the edges $(s,t)$ for all commodities $(s,t)$ or $(s,t,P)$.

\begin{observation}
\label{obs:bounded_target}
	Each \pvar instance with a commodity $(s,t)$ such that $(s,t)\notin\tc(D)$ is a \textup{NO}-instance.
	Apart from these, all instances of \pfixed and \pvar where the target is at least the maximum outdegree of the input graph or the maximum number of commodities starting at the same vertex are \textup{YES}-instances.
	 \end{observation}

\section{Complexity Classification with Respect to the Target}
 \label{sec:target}

Our first---and in a sense introductory---technical section is dedicated to the study of \pfixed\ and \pvar\ when parameterized by the value of the target $T$. 
We begin by noting that both problems of interest are trivially solvable in linear time if $T=0$. 
 Next we consider the case of $T=1$ and have the following tractability result for \pvar.

\begin{theorem}
\label{thm:fixed_target1}
	\pvar is polynomial time solvable when restricted to instances with target $T\leq 1$.
\end{theorem}

\begin{proof}
	Define $\I := (D,K,1)$.
	We assume that for any commodity $(s,t) \in K$, there is an edge $(s,t)$ in $\tc(D)$, because otherwise, we can trivially conclude that $\I$ is a NO-instance.
	
	Our algorithm decides the problem by outputting NO for NO-instances and giving a solution graph to $\I$, otherwise.
	For a directed graph $H$, this algorithm uses a subroutine $\textsc{MergeCycles}(a,b)$, that is defined on two vertices $a$ and $b$ in $V(H)$, where $a$ is in exactly one cycle~$C$ in $H$.
	First, it removes the edges of cycle $C$ from $H$.
	If $b$ is not in a cycle, it adds the edges of an arbitrary cycle that spans $b$ and all vertices in $C$.
	Otherwise, let $b$ be part of a cycle $C'$. Delete the edges of $C'$ and add the edges of a cycle that spans the vertices of both $C$ and $C'$.
	
	The algorithm proceeds in three phases.
	In Phase 1, we start by constructing the graph~$H = (V(D),K)$, that is, the graph on the vertices of $D$ with an edge $(s,t)$ for every commodity $(s,t) \in K$.
	Then for every SCC $S$ of size at least two in $H$, we remove all edges between any two vertices of $S$ and add a cycle that visits all vertices of $S$ in an arbitrary order.
	By the definition of SCCs, it is easy to see that after this step, all cycles in $H$ are exactly the cycles that we just added, and each vertex is contained in at most one cycle.

	In Phase 2, as long as there is a commodity $(s,t)$ and a cycle $C$ in $H$ such that $s \in C$ and $t \notin C$, we select such a commodity. 
	If $(t,s)$ is not an edge in $\tc(D)$, output NO.  
	Otherwise, perform the subroutine $\textsc{MergeCycles}(s,t)$.
	 	
	In Phase 3, let $\thickbar{V}$ be the set of vertices in all cycles in $H$ at the end of Phase 2.
	Then the graph $H - \thickbar{V}$ obtained from $H$ by removing the vertices in $\thickbar{V}$ and their incident edges is an acyclic graph.
	Obtain a topological ordering of the vertices of $H - \thickbar{V}$.
	Let $W$ be a set of \emph{visited} vertices and let initially $W=\emptyset$.
	As long as there is an unvisited vertex in $H - \thickbar{V}$, do the following. 
	Select the earliest vertex $s$ in the ordering such that $s\notin W$ and add $s$ to $W$.
	As long as $s$ has two out-edges $(s,t)$ and $(s,t')$ in $H$, consider the following three cases.
	\begin{itemize}
		\item \textbf{Case 1:} \textit{Both $t$ and $t'$ are in $\thickbar{V}$.}
		If $(t,t')$ and $(t',t)$ are both in $\tc(D)$, then perform $\textsc{MergeCycles}(t,t')$ and remove $(s,t')$.
		Otherwise, output NO.
		\item \textbf{Case 2:} \textit{Exactly one of $t$ and $t'$ is in $\thickbar{V}$.}
		Without loss of generality, suppose $t$ is in $\thickbar{V}$, and $t'$ is not.
		If $(t',t)$ is in $\tc(D)$, then add the edge $(t',t)$ to $H$ and remove $(s,t)$. 
		Otherwise, output NO.
		\item \textbf{Case 3:} \textit{Both $t$ and $t'$ are not in $\thickbar{V}$.}
		Without loss of generality, we assume $t$ precedes $t'$ in the ordering.
		Perform the first subcase that applies:
		\begin{itemize}
			\item \textbf{Subcase 3a:} \textit{Both $(t,t')$ and $(t',t)$ are not in $\tc(D)$.} Output NO.
			\item \textbf{Subcase 3b:} \textit{$(t,t')$ is in $\tc(D)$.} Add the edge $(t,t')$ to $H$ and remove $(s,t')$.
			\item \textbf{Subcase 3c:} \textit{$(t',t)$ is in $\tc(D)$.} Add the edge $(t',t)$ to $H$, and remove $(s,t)$. 
			Compute a new topological ordering of the current $H - \thickbar{V}$ that starts with the vertices in $W$.
		\end{itemize}
	\end{itemize}	

	We first show that when we output NO, $\I$ is indeed a NO-instance, and at any point the execution of the algorithm, the following invariants hold:
	\begin{enumerate}
		\item $H$ is a subgraph of $\tc(D)$.
		\item There is an $s$-$t$-path in $H$, for every commodity $(s,t) \in K$.
		\item If $\I$ is a YES-instance, there is a solution graph $H^*$ such that for every edge $(s,t)$ in $H$, there is an $s$-$t$-path in $H^*$.
		\item Only applicable for Phase 3, $H - \thickbar{V}$ is acyclic and has a topological ordering that starts with the vertices in $W$.
		 	\end{enumerate}
	
	In Phase 1, at initialization, Invariants 1--3 hold by construction and by the definition of SCCs. 
	Note that replacing each SCC by a cycle does not change the validity of these invariants.

	In Phase 2, we prove the invariants by induction with the base case being the graph $H$ obtained right after Phase 1.
   	By the inductive hypothesis, Invariant 3 implies that the vertices in $C$ are part of one SCC in $H^*$ or $\I$ is a NO-instance. Since the maximum outdegree is one, this SCC needs to be a cycle and has no out-edges to vertices outside the SCC. 
	As $(s,t)\in H$, Invariant~3 now implies that $t$ needs to be part of that cycle, which is only possible if $(t,s)$ is an edge in $\tc(D)$. Otherwise we correctly output NO.
	If we merge, Invariant 2 follows immediately and Invariant 3 holds as we just established that the vertices in $C$, $t$, and, if applicable, the previous cycle $C'$ in $H$ that contains $t$ share an SCC in $H^*$.
	Note that $(s,t)\in K$ implies $(s,t)\in E(\tc(D))$, so if $(t,s)$ is an edge in $\tc(D)$, then $s$ and $t$ share an SCC in $D$ and thus the merging preserves Invariant 1.

	For Phase 3, we also prove by induction. Note that Invariant 4 holds in the beginning of Phase 3 by construction. For the induction step, we consider each case separately as follows.
	
	In Case 1, using a similar argument as in Phase 2, we conclude that we have a NO-instance or $t$ and $t'$ have to be in the same cycle in a solution graph $H^*$.
	Hence, if one of $(t,t')$ and $(t',t)$ is not in $\tc(D)$, $\I$ is a NO-instance.
	Otherwise, the vertices in the original cycles containing $t$ and $t'$ are in the same SCC of $D$, and hence, $H$ is still a subgraph of $\tc(D)$ after merging, implying Invariant 1.
 	Invariant 3 holds as all merged vertices share an SCC in $H^*$.
	Since $\textsc{MergeCycles}$ does not reduce reachability and $t'$ can be reached from $s$ via a path to $t$ followed by a path in the cycle that contains both $t$ and $t'$, Invariant 2 holds.
	Since we do not add any edges incident to any vertices in $V(D) \setminus \thickbar{V}$, Invariant 4 still holds.

	In Case 2, Invariant 3 and $T= 1$ imply that we have a NO-instance or a solution graph $H^*$ contains a path from $s$ that visits $t$ and $t'$.
	If this path visits $t'$ before $t$, then $(t',t)$ has to be in $\tc(D)$.
	If it visits $t$ before $t'$, then as argued in Phase 2, $t'$ has to be in the same cycle as $t$, and hence, $(t',t)\in E(\tc(D))$.
	Therefore, if $(t',t)$ is not an edge in $\tc(D)$, then $\I$ is a NO-instance.
	Otherwise, adding the edge of $(t',t)$ keeps $H$ a subgraph of $\tc(D)$ (Invariant 1).
	We remove one edge $(s,t)$ and at the same time add the edge $(t',t)$ to maintain the reachability from $s$ to $t$. 
	Therefore all commodities remain satisfied (Invariant 2).
	The only new edge is $(t',t)$. 
 	We have established that $H^*$ contains a $t'$-$t$-path, as a subpath of an $s$-$t$-path or part of a cycle (Invariant 3).
	 	Further, as $t\in \thickbar{V}$, the added edge $(t',t)$ is not in the graph $H-\thickbar{V}$, so Invariant 4 still holds.
	
	In Case 3, using similar arguments as above, for YES-instances Invariant 3 implies the existence of a solution graph $H^*$ that contains both an $s$-$t$- and an $s$-$t'$-path.
	If both $(t,t')$ and $(t',t)$ are not in $\tc(D)$ no such graph exists and $\I$ is a NO-instance (Subcase 3a).
	In the other subcases, since we only add edges that are in $\tc(D)$ and ensure that reachability is not reduced, Invariants 1 and 2 hold.
	In Subcase 3b, since $t$ precedes $t'$ in the ordering, adding $(t,t')$ keeps $H - \thickbar{V}$ acyclic, and the current ordering remains a topological ordering $H - \thickbar{V}$ (Invariant 4).
	In Subcase 3c, since we skip Subcase 3b, $(t,t')$ is not in $\tc(D)$.
	By Invariant 1, this implies that there is no $t$-$t'$-path in $H$, and hence, after adding $(t',t)$, $H - \thickbar{V}$ is still acyclic.
	Further, no edges to $W$ were added, so the new acyclic graph has a topological sorting that starts with $W$ (Invariant 4).
	For YES-instances, consider the solution graph $H^*$ as before entering Case 3.
	 	In Subcase 3c, the only new edge in $H$ is $(t',t)$. 
	There is a path in $H^*$ that starts at $s$ and visits both $t$ and $t'$.
	As $(t,t')$ is not an edge in $\tc(D)$, this path visits $t'$ before $t$ and thus contains a $t'$-$t$-path as a subpath (Invariant 3).
	 	In Subcase 3b, the only new edge in $H$ is $(t,t')$. If $H^*$ contains a $t$-$t'$-path, Invariant 3 immediately follows.
	Otherwise, by a reasoning analogous to Subcase 3c, $H^*$ contains an $s$-$t$-path that visits $t'$.
	Let $P$ be its $t'$-$t$-subpath. 
	Let $U=u_1,\ldots,u_\ell,t$ be the subsequence of $P$ containing all vertices that precede or equal $t$ in the topological ordering and $U'=t'u_1',\ldots,u_{\ell'}'$ be the remaining subsequence.
	 	We now construct a new solution graph $\tilde{H}^*$ by changing the edges in $H^*$ to move the subsequence $U$ right before $U'$. 
	Formally, all edges stay the same as in $H^*$ apart from the following changes.
	Replace every edge $(v,p)$ with $v\notin P, p\in P$ by an edge $(v,u_1)$. 
	All vertices within the subsequences $U,U'$ lose their out-edge and are instead connected to their successor within the subsequence, if there is any.
	Add the edge  $(t,t')$ to ${\tilde{H}^*}$ and, if there is an edge $(t,v)$ replace it by $(u'_{\ell'},v)$.
 	Because $H^*\subseteq \tc(D)$ and by assumption of the subcase $(t,t')\in \tc(D)$, all the vertices in $P$ share a SCC in $\tc(D)$. 
	Hence, these changes preserve $\tilde{H}^*$ being a subgraph of $\tc(D)$.
	Further note that the changes preserve the maximum outdegree of one. 
	 	Observe that for every $v$-$w$-path in $H^*$ there is a $v$-$w$-path in $\tilde{H}^*$ except if $v\in U'$ and $w\in U$. 
	As then $w$ precedes $v$ in the topological sorting of $H$, there is no edge $(v,w)$ in $H$.
	As further $\tilde{H}^*$ contains the edge $(t,t')$, we have that $\tilde{H}^*$ has a $v$-$w$-path for every edge $(v,w)$ in $H$.
	Combining this insight with Invariant 2 gives that $\tilde{H}^*$ satisfies every commodity in $K$, making $\tilde{H}^*$ a solution graph.
	Hence, $\tilde{H}^*$ witnesses the correctness of Invariant 3.

	Next, note that after Phase 2, all vertices in $\thickbar{V}$ have outdegree one in $H$.
	In Phase 3, we do not increase the outdegree of any vertex in $\thickbar{V}$ and after a vertex in $H \setminus \thickbar{V}$ is processed, it also has outdegree one.
	Hence, when Phase 3 is complete without outputting NO, $H$ has maximum outdegree one.
	Combined with Invariants 1 and 2 above, this implies that $H$ is then a solution graph to $\I$.

	Lastly, we argue that the algorithm terminates.
	At every iteration in Phase 2, we reduce the number of SCCs by one.
	In every case in Phase 3, we remove one out-edge of the current vertex and we do not change or add any out-edges to vertices earlier in the ordering.
	Hence, the algorithm terminates.	
 	This completes the proof of the theorem.
\end{proof}

In contrast to the result above, we show that both \pvar and \pfixed are para\NP-hard when parameterized by $T$. This complements the earlier additive $(+1)$-approximation of Van Dyk, Klause, Koenemann and Megow on orientations of trees with a single source~\cite[Theorem 3]{VanDyk_2023}.

To show this, we reduce from the strongly \nphard \threePart problem \cite{Garey_Johnson_1979}.

\vspace{2mm}
\myproblem{\threePart}
	{\(m,B\in \N\) and \(3m\) positive integers \(n_1,\ldots,n_{3m}\) such that \(\sum_{i=1}^{3m} n_i = mB\) and $\frac{B}{4} < n_i < \frac{B}{2}$ for all $i\in[3m]$.}
	{Can the integers \(n_i\) be partitioned into \(m\) triples such that each triple sums up to \(B\)?}

In particular, \threePart\ remains strongly \NP-hard if all integers $n_1,\ldots,n_{3m}$ are distinct \cite{hulett_multigraph_2008}. 

\begin{theorem}
\label{thm:pvar_pfix_deg_tw}
	\pfixed and \pvar are \nphard, even when restricted to instances where $T\leq 2$, there is a single source, and $\undund{D}$ is a tree of maximum degree at most $4$.	
\end{theorem}

\begin{proof}
	We give a reduction from \threePart with distinct positive integers to \pvar.
	The same construction holds for \pfixed as well by assigning each commodity its unique path in the constructed graph.

	We construct $D$ as follows.
	Let there be a path of $2mB + 2m$ vertices. 
	Let $S$ be the set of the first $2m$ vertices with labels $s_1,s'_1,\ldots,s_m,s'_m$ (in the order they appear in the path).
	The next $mB$ vertices form the set $W$ and are indexed by the numbers $n_1,\ldots, n_{3m}$, specifically $w_1^1,\ldots,w_1^{n_1},\ldots,w_{3m}^1,\ldots,w_{3m}^{n_{3m}}$.
	Among the remaining $mB$ vertices, which form the set $R$, there are $B$ vertices for each triple to be created, that is, $r_1^1,\ldots,r_1^{B},\ldots,r_m^1,\ldots, r_m^B$.	
	For all $i\in [m]$, create the commodities $(s_i,s'_i)$ and $\set{(s_i, r_i^j)\mid j\in [B]}$.
	For all $i\in [3m]$, create the commodities $\set{(w_i^j,w_i^{j+1})\mid j\in [n_i -1]}$.
	For all $i\in [3m]$, create a vertex $\tilde{w}_i$, and add the edge $(w_i^{n_i},\tilde{w}_i)$ and a commodity $(w_i^{n_i},\tilde{w}_i)$.
	For all $i\in [m]$ and $j\in [B]$, create vertices $\tilde{r}_i^j$ and $\thickbar{r}_i^j$, and add the two edges and two commodities $(r_i^j,\tilde{r}_i^j)$ and $(r_i^j,\thickbar{r}_i^j)$.
	The resulting graph $D$ is depicted in \cref{fig:pvar_pfix_deg_tw}.

	\begin{figure}[bt]  		\centering
		\includegraphics{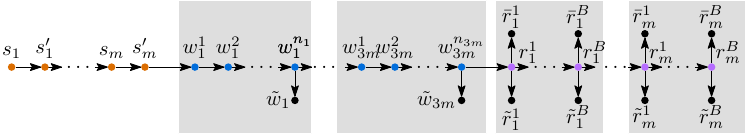}
		\caption{Constructed graph $D$ for the proof of~\cref{thm:pvar_pfix_deg_tw}.
		Vertices in $S$ are colored \textcolor{cb_orange}{orange}, those in $W$ are \textcolor{cb_blue}{blue} and those in $R$ are \textcolor{cb_violet}{violet}.
		}
		\label{fig:pvar_pfix_deg_tw}
		\end{figure}

	Then the \threePart instance is a YES-instance if and only if there is a solution to \pvar in $D$ with target 2.
	Intuitively, for each $i\in[m]$ there are three available edges from $s_i$ and $s'_i$ that can be used to select three subpaths in the middle section, representing three integers. Each subpath $w_j^1,\ldots,w_j^{n_j}$ allows for $n_j$ outgoing edges, thus the subpaths have to be chosen such that their sum is at least $B$ to satisfy the commodities from $s_i$ to the last part of the path.

	Formally, suppose the \threePart instance is a YES-instance. 
	For all $i\in [m]$, let the $i^\text{th}$ triple in the solution be $(n_{\star(i,1)},n_{\star(i,2)},n_{\star(i,3)})$.
	We define
	\begin{align*}
		E_S &:= \set{(s_i,s'_i)\mid i\in[m]},\\
		E_R &:= \set{(r_i^j,\tilde{r}_i^j), (r_i^j,\thickbar{r}_i^j)\mid i\in[m],j\in[B]}.
	\intertext{For all $i\in [3m]$ we define}
		E_W^i &:= \set{(w_i^j,w_i^{j+1})\mid j\in [n_i -1]} \cup \set{(w_i^{n_i},\tilde{w}_i)},
	\intertext{and for all $i\in[m]$ we define}
		E_C^i &:= \set{(s_i, w_{\star(i,1)}^1), (s'_i, w_{\star(i,2)}^1), (s'_i, w_{\star(i,3)}^1)},\\
		E_R^i &:= \set{(w_{\star(i,1)}^j, r_i^j) \mid j\in[n_{\star(i,1)}]} \cup 
				 \set{(w_{\star(i,2)}^j, r_i^{j+n_{\star(i,1)}}) \mid j\in[n_{\star(i,2)}]}\\
				 &\hphantom{= \{} \cup 
				 \set{(w_{\star(i,3)}^j, r_i^{j+n_{\star(i,1)}+n_{\star(i,2)}}) \mid j\in[n_{\star(i,3)}]}.
	\intertext{Finally,}
		E_H &:= E_S \cup E_R \cup \bigcup_{i\in [3m]} E_W^i \cup \bigcup_{i\in [m]} E_C^i \cup \bigcup_{i\in [m]} E_R^i.
	\end{align*} 
	Consider the graph $H$ with edges $E_H$. Then $\delta^+(H) = 2$ as in fact $\delta^+(v) = 2$ for each of the $2mB+2m$ vertices $v$ along the main path and $\delta^+(v) = 0$ for all other vertices.
	Observe that $E_R$ and all $E_W^i$ satisfy all commodities starting in vertices of $R$ and $W$, respectively, and $E_S$ satisfies all commodities with both endpoints in $S$.
	The remaining commodities start at the vertices $s_1,\ldots,s_m$ and each lead to a set of $B$ vertices in $R$, such that the sets are pairwise disjoint.
	For $i\in [m]$, note that $s_i$, by also using $s'_i$, is connected to the first vertex of three subpaths of vertices in $W$, that the total length of these subpaths is exactly $B$, and that each of the vertices in the subpaths connects to a distinct destination required for $s_i$. 
	Thus, these remaining commodities are satisfied as well, and $H$ is a solution to the \pvar instance.
	
	Now suppose there is a solution graph $H$ for the \pvar instance with target 2.
	First note that for all $i\in[3m]$ we have $E_S, E_R, E_W^i \subseteq E(H)$, as otherwise there would be unsatisfied commodities.
	By $E_R \subseteq E(H)$ alone, we have that each vertex in $R$ has outdegree 2, and they can have no other edges in $E(H)$. 
	As each $t\in R$ is the destination of one commodity, this implies that there have to be $|R|=mB$ individual edges from $V\setminus R$ to $R$ in $E(H)$. 
	By construction of $D$, these edges have to start in $S$ or $W$ as otherwise they would not be in $\tc(D)$.
	Each of the $mB + 2m$ vertices in $S\cup W$ is allowed 2 out-edges, out of which $E_S$ and the $E_W^i$ already use $mB + m$. 
	This leaves another $mB + 3m$ edges, out of which $mB$ have to go to $R$.
	We will show below that the remaining $3m$ edges have to end in the first vertex in each of the $3m$ subpaths in $W$, that is, in each of $w_1^1,\ldots,w_{3m}^1$.
	In total, the entire available outdegree of $S$ and $W$ is accounted for. 

	Suppose one of the $3m$ edges would not connect to a subpath. This would save 1 edge that can be connected to $R$, but now there is $j\in[3m]$ such that there is no connection from $S$ to any $w_j^1,\ldots,w_j^{n_j}$. As $n_j$ of the edges in our calculation start in $w_j^1,\ldots,w_j^{n_j}$ and can then not be used for fulfilling commodities, this would reduce the total number of available edges by $n_j - 1$. As $n_j > \frac{B}{4}$ and we can assume $B\ge 4$ (otherwise the \threePart instance becomes trivial), this reduces the amount of available edges by at least 1 and there are no longer enough edges to satisfy all commodities ending in $R$.
	The same issue arises when connecting to some vertex $w_i^j$ instead of $w_i^1$ for some $i\in[3m], j > 1$.
	Hence, the remaining $3m$ edges end in the vertices $w_1^1,\ldots,w_{3m}^1$.
	In particular, this implies that there are no $a,b\in [m], a\neq b$ such that $s_a$ and $s_b$ are connected to the same subpath $w_j^1,\ldots,w_j^{n_j}$ in $H$. 

	We now argue that $s_i$ is connected to exactly three subpaths $w_j^1,\ldots,w_j^{n_j}$ in $H$ and that the total length of these subpaths is $B$. 
	We then have a solution to \threePart where the $i^\text{th}$ triple is $n_a,n_b,n_c$ where $s_i$ is (not necessarily directly) connected to $w_a^1,w_b^1,$ and $w_c^1$. 
	Note that we just established that these triples are pairwise disjoint as no two vertices $s_a,s_b,a\neq b$ connect to the same subpath.
	Consider any $i\in [m]$ and the remaining three outgoing edges of $s_i$ and $s'_i$.
	As argued above, they either end in $R$ or distinct vertices in $W_1 = \set{w_1^1,\ldots,w_{3m}^1}$.
	If $s_i$ is connected to no subpath, the three edges cannot cover all commodities of $s_i$.
	Thus, there is $a\in [3m]$ such that $(s_i, w_a^1)$ or $(s'_i,w_a^1)$ is in $E(H)$.
	As this used one edge and $w_a^1,\ldots,w_a^{n_a}$ has $n_a$ outgoing edges, we now have $n_a +2$ edges available. As we still assume $B \ge 4$ and have $n_a < \frac{B}{2}$, this does not suffice.
	Thus, there is $b\in [3m]$ such that one of the $n_a +2$ available edges ends in $w_b^1$.
	Now, there are $n_a + n_b + 1$ edges available. 
	However, recall that $n_a\neq n_b$, so we have $n_a + n_b -1 < \frac{B}{2} + \frac{B}{2} - 1 + 1 = B$, so these edges still do not suffice.
	Thus, there is $c\in [3m]$ such that one of the $n_a +n_b+1$ available edges ends in $w_c^1$.
	This gives $n_a + n_b + n_c$ edges available for connecting to $r_i^1,\ldots,r_i^B$.
	As for all $i\in m$ we have that $s_i$ is connected to at least 3 distinct vertices in $W_1$ and $|W_1|=3m$ we have that there is no $s_i$ that is connected to more than 3 vertices in $W_1$.
	Hence, each $s_i$ is associated with a triple $n_a,n_b,n_c$, has exactly $n_a+n_b+n_c$ edges available to connect to its $B$ distinct destinations in $R$.
	Thus, for each $s_i$ we have $n_a+n_b+n_c \ge B$, and as $\sum_{j\in[3m]} n_j = mB$ we have $n_a + n_b + n_c = B$, yielding a solution to the \threePart instance.
\end{proof} 

\section{Parameterizing by the Number of Commodities}
\label{sec:commodities}

This section is dedicated to the study of the considered parcel sortation problems when parameterized by the number of commodities. This line of enquiry can be seen as analogous to the fundamental questions that have been investigated for other prominent examples of problems typical for logistical networks, such as the study of \textsc{Vertex} and \textsc{Edge Disjoint Paths} parameterized by the number of paths or of \textsc{Coordinated Motion Planning} parameterized by the number of robots. We remark that the former has been shown to be fixed-parameter tractable in Robertson and Seymour's seminal work~\cite{RobertsonS95b}, while the fixed-parameter tractability of the latter has only been established on planar grid graphs in a recent work of Eiben, Ganian and Kanj~\cite{EibenGK23}.

We start by bounding the size of a hypothetical solution by a function of the number of commodities---an insight which will be used later on in this section.
\begin{lemma}
\label{lem:commodities_bounded_sol}
	Every YES-instance $(D,K,T)$ of \pfixed or \pvar has a solution path cover whose unions consists of at most $2^{|K|} + |K|$ vertices.
\end{lemma}
\begin{proof}
	Let $k := |K|$ and $K = \set{\kappa_1, \dots, \kappa_k}$.
	By deleting excess vertices and edges from any solution, there is a solution path cover with the paths $Q_1, \dots, Q_k$, where $Q_i$ is a path that satisfies $\kappa_i$ for $i \in [k]$. Let $H$ be the solution graph that is the union of the solution path cover.
	For every non-source vertex $v\in V(H)$, we define $\calQ(v) := \set{i \mid v \in Q_i}$.
	In particular, note that for every non-source vertex $v$ in the graph, we have $\calQ(v) \neq \emptyset$.

	Suppose there are non-source vertices $u,v$ such that $\calQ(u) = \calQ(v)$.
	Consider the following modification of $H$.
	For every path $Q_i$ that visits $u$ before $v$, we remove the subpath between $u$ (inclusive) and $v$ (exclusive) from $Q_i$.
	In other words, if $Q_i = (v_1, \dots, v_{\ell}, u, v_{\ell+1}, \dots, v_ {j}, v, \allowbreak  v_{j+1}, \dots, v_t)$, then we replace this path by $(v_1, \dots, v_{\ell}, v, v_{j+1}, \dots, v_t)$.
	Similarly, for every path $Q_i$ that visits $v$ before $u$, we remove the subpath between $v$ (inclusive) and $u$ (exclusive) from $Q_i$.
	Note that after the modification, the two vertices $u$ and $v$ no longer satisfy $\calQ(u) = \calQ(v)$.

	We argue that the graph $H'$ as the union of all $k$ paths after the modification is still a solution graph to $(D,K,T)$.
	It is easy to see that every modified path is in the transitive closure of the original path, and hence the modified path still satisfies the corresponding commodity.
	It remains to show that the maximum outdegree does not increase.

	Observe that every newly added edge ends in $u$ or $v$ and does not start at either $u$ nor $v$.
	Further, note that if the edge $(a,u)$ was added, then previously $(a,v)\in E(H)$ and vice versa.
	Hence, if $(a,u),(a,v)\in E(H')$ then $(a,v),(a,u)\in E(H)$ and the degree of $a$ did not increase.
	Otherwise, let without loss of generality $(a,v)\in E(H')$ and $(a,u)\in E(H)$.
	As before the modification $\calQ(u) = \calQ(v)$, the edge $(a,u)$ in $H$ is only used in paths that visit both $u$ and $v$. 
	Consider any such path. 
	If it visits $u$ before $v$ in $H$, the path skips $u$ in $H'$ and hence does not contain the edge $(a,u)$.
	If it visits $v$ before $u$ in $H$, then using the edge $(a,u)$ would imply that $a$ lies in between $v$ and $u$ in the path. Thus, $a$ is skipped in $H'$ and the path does not contain $(a,u)$.
	Hence, $(a,u)\notin H'$, which compensates for the additional edge $(a,v)$ and $a$ does not increase its degree.	
	
	Apply the modification above repeatedly as long as there are two non-source vertices $u,v$ such that $\calQ(u) = \calQ(v)$.
	This procedure terminates after less than $k\cdot |V(D)|$ iterations as each time at least one of the $Q_i$ paths is shortened.  
	At the end, we obtain a solution graph such that every non-source vertex has a distinct $\calQ$ value.
	The lemma then follows from the facts that (i) the codomain of $\calQ$ is the set of all subsets of $[k]$ and (ii) there are at most $k$ sources.
\end{proof}
  
\subsection{A Fixed-Parameter Algorithm for \pfixed}
In this subsection, we establish the fixed-parameter tractability of \pfixed\ by the number of commodities.
 Our proof relies on three main ingredients.
First, we provide a straightforward proof that \pfixed instances with target $T \leq 1$ are in \fpt by the number of commodities (\cref{thm:var_commodities_target1}). This allows us to concentrate on the case of $T\geq 2$ in the rest of the section. 

\begin{theorem}
\label{thm:var_commodities_target1} 
	\pfixed restricted to instances such that $T \leq 1$ is in \fpt w.r.t.\ the number of commodities.
\end{theorem}
    \begin{proof}
 	We prove the statement by providing an $f(|K|)$-kernel. 
 	In particular, we show that any \pfixed instance $(D,K,1)$ can be reduced to the equivalent instance $(D',K',1)$, where $V(D')$ contains only vertices that are the source or destination of a commodity in $K$, $E(D') = \{(v,v')\mid v,v'\in V(D'), (v,v')\in\tc(D)\}$, and $K'$ is obtained from $K$ by deleting all vertices from all paths that are neither start nor destination of any commodity.

	As $\tc(D')$ is a subgraph of $\tc(D)$, any solution to $\I'=(D',K',1)$ is a solution to $\I=(D,K,1)$.
	In the reversed direction, we show that if $\I$ is a YES-instance then it has a solution that does not use any of the deleted vertices by an exchange argument.
	Let $H$ be a solution graph for $\I$ and let $v$ be a non-source, non-destination vertex in $H$. 
	If $v$ has no out-neighbor in $H$, it does not contribute to the solution and can be removed.
	Otherwise, as $H$ has outdegree 1, $v$ has a unique out-neighbor $w$ in $H$. 
	Create a new solution graph $H'$ by letting $V(H') = V(H)\setminus \set{v}$, deleting the edge $(v,w)$, and replacing each edge $(u,v)\in E(H)$ by an edge $(u,w)\in E(H')$. 
	Then the outdegree of each vertex in $H'$ is at most as large as its outdegree in $H$ and thus at most 1.
 	Moreover, for every commodity routed along an $s$-$t$-path $P$ there is still an $s$-$t$-path in $H'\cap \tc(P)$. 
	In either case, $H'$ still resolves all commodities in $K$.
	By repeatedly applying this procedure, we can create a solution graph $H'$ in which all non-source, non-destination vertices are removed. 
	Thereby $H'$ is a solution graph for $\I'$.

	Clearly, the above reduction is computable in polynomial time and leaves a graph with at most $2|K|$ vertices.
	By brute forcing the reduced instance, this kernel yields \fpt runtime. 
\end{proof} 

Second, we recall the following generalization of Ramsey's Theorem \cite{ramsey}:

\begin{fact}
\label{thm:ramsey}
There is a function \ram: $\N \times \N \to \N$ with the following property.
For each pair of integers $r, s$ and each clique $G$ of size at least \ram$(r, s)$ such that each edge has a single label (color) from a set of $r$ possible labels, it holds that $G$ has a subclique $H$ of size at least $s$ such that every edge in $H$ has the same label. 
\end{fact}

In particular, we remark that \ram\ has a computable upper bound.

Third, we define the following generalization of \pfixed.

  \vspace{0.2cm}
\myproblem{\textsc{Subset-Min-Degree Sort Point Problem} (\pfixedgen)}
	{A digraph $D=(V,E)$, a set $F \subseteq V$ of vertices, a target $T\in \N$, and a set $K$ of routed commodities each of which is a tuple of the form $(s,t,P)$ where $s, t \in V$ and $P$ is an $s$-$t$-path in $D$.}	
	{Is there a subgraph $H$ of the transitive closure $\tc(D)$ of $D$ such that every vertex in $V(D) \setminus F$ has outdegree at most $T$ in $H$, and for every commodity $(s,t,P) \in K$, there exists a directed $s$-$t$-path in $H \cap \tc(P)$?
}

  \vspace{0.2cm}
We call the vertices in $F$ in the definition above \emph{flexible}, as their outdegrees are not constrained.
The terms solution graph and solution path cover are defined analogously for \pfixedgen.
The size of a solution can be bounded by a function of the number of commodities using an identical argument as in the proof of \cref{lem:commodities_bounded_sol}. 

\begin{lemma}\label{lem:commodities_bounded_sol_gen}
	Every YES-instance $(D,K,T)$ of \pfixedgen has a solution path cover whose union consists of at most $2^{|K|} + |K|$ vertices.
\end{lemma}

For an \pfixed instance $(D,K,T)$ and $v\in V(D)$, let the \emph{type} $\calP(v)$ of $v$ be the set of routed commodities with paths using $v$ (i.e.,  $\calP(v) = \{(s,t,P)\mid (s,t,P)\in K,\allowbreak P$~contains~$v\}$). 
In \cref{lem:comm_fixed_equiv_subset} below, we reduce \pfixed instances to \pfixedgen with a nontrivial set $F$ such that there are only few non-flexible vertices. To define $F$, we build on the following definition.

\begin{definition}[$q$-enclosure]
	Let $\I = (D,K,T)$ be an \pfixed instance and $q$ be a natural number.
	A vertex $v$ is \emph{$q$-enclosed} w.r.t.\ $\I$, if there exist $U\subseteq K$, $U \neq \emptyset$, and distinct vertices $v_1, \dots, v_{2q+1}$ of type $U$ in $D$ such that $v_{q+1} = v$, and for all commodities $(s, t, P) \in U$, the directed path $P$ contains $(v_1, \dots, v_{2q+1})$ or $(v_{2q+1}, \dots, v_1)$ as a subsequence.
	 \end{definition}

\begin{lemma}
 \label{lem:comm_fixed_equiv_subset}
	For an \pfixed instance $(D,K,T)$ with $T \geq 2$, let $k:=|K|$, and $q = \ceil{(2^k+k)\big(1+\frac{k}{T-1}\big)}$.
	Suppose that $|V(D)| > q$.
	Further suppose $F$ is a set of $q$-enclosed vertices w.r.t.\ $(D,K,T)$.
	Then $(D,K,T)$ is a YES-instance if and only if the \pfixedgen instance $(D,F,K,T)$ is a YES-instance.
\end{lemma}

\begin{proof}
	Trivially, a solution path cover of $(D,K,T)$ is also a solution path cover of $(D,F,K,T)$.
	Hence, we only need to prove the other direction.

	Suppose that $C$ is a solution path cover of $(D,F,K,T)$.
	By \cref{lem:commodities_bounded_sol_gen}, we can assume that the union of $C$ has at most $2^k + k$ vertices.
	Let $X$ be the set of vertices not in $C$.
	Note that $C$ is also a solution path cover of $(D, F \setminus X, K,T)$.
	
	We prove by induction that there exists a sequence $C_0, C_1, \dots, C_{|F \setminus X|}$ such that for $i \in \set{0, \dots, |F \setminus X|}$, $C_i$ has at most $2^k + k+ki/(T-1)$ distinct vertices, and $C_i$ is a solution path cover of $(D, F_i, K,T)$, where $|F_i| = |F \setminus X| - i$ and $F_i \subseteq F \setminus X$.
	
	It is easy to see that the choice of $C_0 \equiv C$ satisfies the requirement.
	For the inductive step, suppose that for some value of $i < |F \setminus X|$, there exist $C_0, \dots, C_i$ as required.
	Let $v$ be a vertex in $F_i$.
	From $C_i$, we construct a solution path cover $C_{i+1}$ of the \pfixedgen instance $(D, F_i \setminus \set{v},K, T)$ such that $C_{i+1}$ has at most $2^k + k + k(i+1)/(T-1)$ distinct vertices.

	Since $F_i \subseteq F$, $v$ is $q$-enclosed.
	Let $v_1, \dots, v_{2q+1}$ be vertices of some type $U \subseteq K$ that witness its $q$-enclosure.
	In particular, $v_{q+1} = v$.
	By definition, $v_1, \dots, v_{2q+1}$ are only used in the paths in $U$, and these paths visit these vertices in the order $v_1,\dots, v_{2q+1}$ (right commodities) or $v_{2q+1},\dots,v_1$ (left commodities).

	We construct $C_{i+1}$ from $C_i$ in three phases.
	In the first phase, for each commodity whose witness path in $C_i$ does not visit $v$, we add the witness path to $C_{i+1}$.
	Let $H'$ be the union of the paths in $C_{i+1}$.
	Then for each remaining commodity $(s,t,P)$, we add to $H'$ the $s$-$v$-subpath of its witness path in $C_i$.
	Up to this point, the outdegree of $v$ in $H'$ is zero and all unsatisfied commodities in $H'$ have a directed path from their source to $v$.

	In the second phase, we process the set $R$ of the right commodities that do not yet have a witness path in $C_{i+1}$.
 	If $R=\emptyset$, there is nothing to do.
	Otherwise, for $j \in \set{q+2, \dots, 2q+1}$, define $r(j) := T - \delta^+_{H'}(v_j)$.
	We iteratively resolve all commodities in $R$ and initially let $\ell := q+1$.
	We add the edge $(v_\ell,v_p)$ to $H'$, where $p$ is the minimal index such that $\ell < p \le 2q+1$ and $r(p) > 0$.
	If $|R| > r(p)$, choose arbitrary $r(p)-1$ right commodities in $R$, otherwise, choose all of them.
	For each chosen commodity $(s,t,P)$, we add the edge $(v_p,t)$ to $H'$ and remove the commodity from $R$.
	Additionally, we add its witness path to $C_{i+1}$ as follows: The path is the concatenation of the $s$-$v$-subpath of its witness path in $C_i$, the $v$-$v_p$-path in $H'$ using only vertices in $\set{v_1, \dots, v_{2q+1}}$, and the newly added edge $(v_p,t)$.	
	After processing these chosen commodities, if $R\neq \emptyset$,
	 	 we let $\ell = p$ and proceed with the next iteration.

	In the last phase, we process the left commodities that do not yet have a witness path in $C_{i+1}$ in a similar manner, except that we iterate backwards from $q+1$ to $1$.

	We first show that $C_{i+1}$ is a solution path cover of $(D, F_i \setminus \set{v},K, T)$.
	Note that we add at most one out-edge from $v$ in each of the last two phases, and hence $\delta^+_{H'}(v) \leq 2 \leq T$.
	Let $H$ be the union of $C_i$, that is, $H$ is a solution graph for $(D, F_i, K, T)$.
	Then, by construction, for all vertices $v'\neq v$ in $V(H')$ we have  $\delta^+_{H'}(v') \leq \delta^+_{H}(v')$ or $\delta^+_{H'} \le T$. 
	Since $H$ is a solution graph of $(D, F_i,K, T)$, it follows that the outdegree in $H'$ of a vertex outside the set $F_i \setminus \set{v}$ is at most $T$.
	It remains to show that each commodity in $K$ has a witness path in $C_{i+1}$.
	By construction and by symmetry, we only need to prove that there exists $p$ in the second phase, and by the end of that phase, the set $R$ is empty.
	By the inductive hypothesis, the number of distinct vertices in $C_i$ is at most $2^k + k + ki/(T-1)$. 
	Combined with the definition of $q$ and the fact that $i < |F \setminus X| \leq |V(D) \setminus X| \leq 2^k+k$, this implies that there are at least $k/(T-1)$ vertices among $v_{q+2}, \dots, v_{2q+1}$ that are not used in $C_i$.
	Hence, the minimal index $p$ as required in the second phase exists.
	Moreover, since there are at most $k$ right-commodities in $K$, it is easy to see that these vertices are sufficient to ensure that $R$ is empty by the end of the second phase.
		
	Finally, since we process at most a total of $k$ commodities in the second and third phases, it is easy to see that we added out-edges to at most $k/(T-1)$ vertices in $H$. 
	Thus, the number of distinct vertices in $C_{i+1}$ increases by at most $k/(T-1)$ compared to $C_i$.
	Using the inductive hypothesis, the number of distinct vertices in $C_{i+1}$ hence is bounded by $2^k + k + k(i+1)/(T-1)$ as required. 
	
	Note that $|F_{|F \setminus X|}| = 0$ so $(D, F_{|F \setminus X|}, K,T)$ is equivalent to $(D,K,T)$.
	Hence, $C_{|F \setminus X|}$ is a solution path cover of $(D,K,T)$.
\end{proof}

We now have all components to establish the main result of this subsection.
 
\begin{theorem}
 \label{thm:fixed_commodities}
	\pfixed is in \fpt w.r.t.\ the number of commodities.
\end{theorem}

\begin{proof}
	Let $(D, K, T)$ be an \pfixed instance with $|K| = k$ and $|V(D)| = n$.
	When $T \le 1$, we have fixed-parameter tractability by \cref{thm:var_commodities_target1}.
	Therefore, we now assume that $T \geq 2$.
	Recall the \ram\ function from \cref{thm:ramsey}.
	Let $\overline{\ram}: \N \times \N \rightarrow \N$ be a computable upper bound of \ram.

	For any vertex $v$ such that $\calP(v) = \emptyset$, we can remove it and its incident edges from $D$, because no solution will use $v$.	
	Let $q := \ceil{(2^k+k)\big(1+\frac{k}{T-1}\big)}$.
	If $n \leq q$, a solution graph of $(D,K,T)$, if one exists, can be brute forced in \fpt time.
	Hence, we now assume that $n > q$.
	
	For any nonempty subset $U$ of $K$, let $V_U$ be the set of vertices $v$ such that $\calP(v) = U$.	
	Further, let $(s_1, t_1, P_1), \dots, (s_{|U|}, t_{|U|}, P_{|U|})$ be the commodities in $U$.
	Define an edge-labeled complete graph on $V_U$, where the label of each edge is a $(|U|-1)$-tuple, whose $i^\text{th}$ element is a binary number to indicate whether $P_{i+1}$ visits the incident vertices of the edge in the same or opposite order as $P_1$.
	(Note that when $|U| = 1$, there is only one label.)
	As long as $|V_U| \geq \overline{\ram}(2^{|U|-1},2q+1)$, we choose an arbitrary subset of $V_U$ of size $\overline{\ram}(2^{|U|-1},2q+1)$, and by \cref{thm:ramsey}, we can find a subclique of size $2q+1$ whose edges have the same label.
	In other words, the paths in $U$ visit the vertices in this subclique either in the same order as $P_1$ does or in the reversed order.
	Then we mark the $(q+1)^{\text{th}}$ vertex in this order as a $q$-enclosed vertex.
	We repeat this procedure on the unmarked vertices until we obtain an instance where the number of unmarked vertices in $V_U$ is smaller than $\overline{\ram}(2^{k-1},2q+1)$. 
	
	After processing all subsets $U$ of $K$, we obtain an instance whose input graph has at most $f(k) := 2^k \cdot \overline{\ram}(2^{k-1},2q+1)$ unmarked vertices, while the marked vertices are $q$-enclosed.
	Let $F$ be the set of all marked vertices.
	By \cref{lem:comm_fixed_equiv_subset}, $(D,K,T)$ is a YES-instance, if and only if the \pfixedgen instance $(D,F,K,T)$ is so.

	We can assume that each edge in $D$ appears in the path of some commodity in $K$, as other edges are obsolete and can be removed.
	Hence, every unmarked vertex has outdegree at most $k$.
	Let $\thickbar{E}$ be the set of out-edges of all unmarked vertices in $D$.
	Let $\tilde{E}$ be the set of edges from every marked vertex~$v$ in $F$ to every destination in $K$ that is reachable from $v$ in $D$.
	Recall that in a solution graph of $(D,F,K,T)$, a vertex in $F$ can have arbitrarily large outdegree.	
	Therefore, if $(D,F,K,T)$ is a YES-instance, there exists a solution graph that contains all edges of $\tilde{E}$.
	We can use a brute force approach to find such a solution graph as follows.
	Consider the graph $D' := (V(D), \thickbar{E})$.
	For each $E^+\subseteq E(\tc(D'))$, we check if the graph $H=(V(D), E^+ \cup \tilde{E})$ is a solution graph of $(D,F,K,T)$.
	That way, we can decide $(D,F,K,T)$ and, as discussed above, decide $(D,K,T)$.

	Finally, we analyze the run time.
	Since after every application of \cref{thm:ramsey}, we mark one vertex, in total, we apply \cref{thm:ramsey} at most $n$ times.
	To find the subclique guaranteed by \cref{thm:ramsey}, we can use a brute force approach of trying all $(2q+1)$-vertex subsets of the vertex set of size $\overline{\ram}(2^{|U|-1},2q+1)$.
	This takes time $\bigO{\overline{\ram}(2^{k-1},2k+1)^{(2k+1)}}$.
	Next, since there are at most $f(k)$ unmarked vertices, each of which has outdegree at most $k$, we have $|\thickbar{E}| \leq kf(k)$.
	Therefore, the number of subsets of $E(\tc(D'))$ is bounded by some function of $k$.
	In total, the whole algorithm runs in \fpt time.
\end{proof}

\subsection{A Fixed-Parameter Algorithm for \pvar}
We use the color-coding technique with derandomization~\cite[Subsections 5.2 and 5.6]{CyganFKLMPPS15} to establish fixed-parameter tractability of \pvar with respect to the number of commodities.
Intuitively, we solve a restricted version of \pvar, namely \textsc{COLORFUL-}\pvar, where the vertices in $D$ are colored and a solution graph $H$ is not allowed to contain two vertices of the same color.
By \Cref{lem:commodities_bounded_sol}, this requires $2^{|K|}+|K|$ distinct colors. 
Given an \FPT algorithm for the colorful variant we obtain an \FPT algorithm for the original \pvar problem by applying the \FPT algorithm on different colorings of the input graph $D$. 
The color-coding technique establishes that with $\gamma$ colors we do not need to try more than $f(\gamma)\log(n)$ colorings to be guaranteed that any set of at most $\gamma$ vertices (and thereby any possible solution of that size) is colorful in at least one of the colorings, meaning that all vertices in the set have distinct colors.

We note that the algorithm underlying Theorem~\ref{thm:fixed_commodities} cannot be applied to Theorem~\ref{thm:pvar_commodities_fpt}, as it hinges on identifying ``irrelevant'' ($q$-enclosed) vertices. The crucial criterion for this relevance is the way in which the routed commodities use these vertices in their paths---some information which is not present for instances of \pvar.  On the other hand, the color-coding based algorithm for the theorem below does not trivially adapt to \pfixed as it requires to rearrange the subgraphs of a solution induced by SCCs into paths, which does not violate unrouted commodities but might violate routed ones.

\begin{theorem}
 \label{thm:pvar_commodities_fpt}
	\pvar is in \fpt w.r.t.\ the number of commodities.
\end{theorem}
\begin{proof}
	Let $(D,K,T)$ be an \pvar instance and $k := |K|$.
	If $T\le 1$, we solve the instance using \cref{thm:fixed_target1}.
	Hence, we can assume that $T \geq 2$.
	We start by giving an algorithm to solve the colorful variant, where the vertices in $D$ are colored by $2^k+k$ colors and any solution can use each color at most one time. Formally, a coloring of a set $V$ of vertices is a function $c\colon V \rightarrow [2^k+k]$, so instances for \textsc{COLORFUL-}\pvar are of the form $(D,K,T,c)$.
	
	Let $S$ and $R$ be the set of sources of commodities in $K$ and the set of sources or destinations of commodities in $K$, respectively (i.e., $S = \set{s\mid (s,t)\in K}$ and $R = S\cup \set{t\mid (s,t)\in K}$). 
	We now prove that any colorful YES-instance has a solution graph $H'$ in which each SCC either has size one or contains a vertex in $S$.
	Let $H$ be any solution graph for $(D,K,T,c)$ with an SCC $C$ without a vertex in $S$ and $|C|\ge 2$.
	Consider the solution graph $H'$ obtained from $H$ by replacing $H[C]$ by a path of the vertices in $C$ in an arbitrary order and all edges from $V(D)\setminus C$ to $C$ are redirected to end in the first vertex of the path. Edges from $C$ to $V(D)\setminus S$ remain unchanged. 
	In the original SCC, each vertex has at least one outgoing edge inside the SCC, so this modification does not increase the outdegree of any vertex.
	Further, $H'$ satisfies each commodity $(s,t)$: 
	For a commodity $(s,t \in K)$, if $H$ has an $s$-$t$-path that did not enter $S$, then the same path exists in $H'$. 
	Otherwise, since $s\in S$, $s\notin C$. 
	Thus there is a path from $s$ to the first vertex in the newly created path and hence a path from $s$ to each vertex in $C$. 
	Thereby, $s$ can reach every vertex that it could reach in $H$. 
	Hence, $H'$ is a solution graph to $(D,K,T,c)$ by additionally noting that $H'$ is colorful since $V(H')=V(H)$.
	As $H'$ did not introduce any new SCC of size at least two, using the above procedure exhaustively, we obtain a solution graph where each SCC of size at least two contains a vertex in $S$.

	We construct an algorithm that works on all colorful instances $(D,K,T,c)$ that require $c$ to assign a unique color to each vertex $v$ in $R$ (i.e., there are no other vertices in $D$ that share the same color as $v$).  	We now branch on all possible \emph{templates} for a solution. 
	Formally, a template $\Gamma$ is a graph of up to $2^k+k$ vertices such that $R \subseteq V(\Gamma)$, all vertices in $V(\Gamma)$ have a unique color from $[2^k+k]$ and vertices in $R$ have the same color in $\Gamma$ as in $D$.
	Further, we only consider templates that are the union of one $s$-$t$-path per commodity $(s,t)\in K$, which implies that each vertex has outdegree at most $k$.
	There are fewer than $\binom{2^k+k}{k}^{2^k+k} \cdot (2^k+k)!$ many such templates.
	Discard a template $\Gamma$ if 
	\begin{itemize}
		\item its maximum outdegree is larger than $T$,
		\item it contains an SCC of size at least 2 without a vertex from $S$, or
		\item there are $v,w\in R$ such that $(v,w)\in E(\Gamma)$ but $(v,w)\notin \tc(D)$. 	\end{itemize} 
	For every remaining template $\Gamma$, we verify whether it yields a solution to the colorful $(D,K,T,c)$ instance:
	We mark all vertices in $R$ as \emph{viable}.
	Consider a topological sorting of the SCCs in $\Gamma$.
	Construct a total order of the vertices in $V(D)\setminus R$ according to a linearization of that topological sorting, where vertices inside an SCC of size at least 2 are placed in arbitrary order at the position of the SCC.
	We visit the vertices in that order from back to front; that is, we start with a vertex with no out-edges except inside its SCC or to $R$.

	Let the currently visited vertex $v'$ in $\Gamma$ have color $a$. 
	We mark every vertex $v$ in $D$ of that color as \emph{viable} for which $v$ satisfies the following requirements: 
\begin{itemize}
	\item If $v'$ is in an SCC $C$ of size at least 2, let $s$ be any vertex in $C\cap S$. We require $(v,s)$ and $(s,v)$ to be edges in $\tc(D)$. 
	\item Otherwise, for each edge $(v',w')\in E(\Gamma)$ we require that there is a viable vertex $w$ in $D$ of the same color as $w'$ in $\Gamma$ such that the edge $(v,w)$ is in $\tc(D)$.
	For each edge $(w',v')\in E(\Gamma)$ with $w'\in R$ we additionally require $(w',v)\in\tc(D)$.
\end{itemize}
	Note that in the sorting, all edges between vertices that are not part of the same SCC go from left to right. Hence, whenever we look for a viable vertex $w$, viability for all vertices with the color of $w'$ has already been decided.
	We have a colorful YES-instance if and only if for some template $\Gamma$ there is a viable vertex for each color used in $\Gamma$.

	\begin{figure}
		\includegraphics[page=1,scale=1.2]{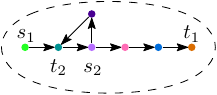}
		\hspace{2cm}
		\includegraphics[page=6,scale=1.2]{color_coding_algo.pdf}\\
		\bigskip
		 \includegraphics[page=2,scale=1.2]{color_coding_algo.pdf}
		\hfill
		\includegraphics[page=3,scale=1.2]{color_coding_algo.pdf}
		\hfill
		\includegraphics[page=4,scale=1.2]{color_coding_algo.pdf}
		\hfill
		\includegraphics[page=5,scale=1.2]{color_coding_algo.pdf}
		\caption{Visualization of the algorithm for \cref{thm:pvar_commodities_fpt}.
		A template $\Gamma$ is displayed in the top-left. The bottom row shows how vertices in $D$ are marked as viable (square) or not viable (cross). First, only sources and destinations are marked as viable (left). Only the \textcolor{cb_blue}{blue} vertex $e$ has a path to the \textcolor{cb_orange}{orange} vertex $t_1$ and is marked as viable (center-left). The \textcolor{cb_pink}{pink} vertex $d$ is viable as it is reachable from $s_2$ and has a path to the viable \textcolor{cb_blue}{blue} vertex $e$ (center-right). As both \textcolor{cb_purple}{purple} vertices $b$ and $c$ share an SCC with $s_2$ in both $\gamma$ and $D$, both are viable (right).
		Choosing out-edges according to the template might result in a solution graph $H$ as displayed in the top-right. Note that both viable \textcolor{cb_purple}{purple} vertices have an edge to $t_2$ but only (arbitrary) one of them has $s_2$ as an in-neighbor.}
		\label{fig:color_coding_rspp_by_comm}
	\end{figure}

	As typical for the color-coding technique~\cite[Subsections 5.2 and 5.6]{CyganFKLMPPS15}, this algorithm is used on each member of a suitable family of colorings for $D$.
	Here, as the vertices in $R$ need to be part of every solution, we first fix the colors of vertices in $R$ arbitrarily (but without repetition) in all colorings. The members of the family then only differ in the way they assign the remaining $2^k+k-|R|$ colors to the remaining vertices.
	The original uncolored instance $(D,K,T)$ is a YES-instance if and only if any member of the family yields a YES-instance for the colorful instance.
	This approach takes $2^{O(2^k k)} n^{O(1)}$ time in total. 

	Suppose the algorithm decides for a YES-instance witnessed by a certain template $\Gamma$. 
	We construct a colorful solution graph $H$ as follows and note that even though it is build on $\Gamma$, it may diverge from the structure proposed by the template.
	Initially, $H$ is empty.
	Add all viable vertices of every color to $H$.
	Iterate over the vertices in $\Gamma$. 
	For the current vertex $v'$, let $a$ be its color.
	For each $(v',w')\in E(\Gamma)$ add edges as follows.
	Let $a_{w'}$ be the color of $w'$, and for each viable vertex $v$ of color $a$ add the edge $(v,w)$ to $H$, where $w$ is an arbitrary viable vertex of color $a_{w'}$ such that $(v,w)$ is an edge in $\tc(D)$.
	Such a vertex always exists for the following reasons.
	\begin{itemize}
	\item If $v'$ and $w'$ are in the same SCC $C$ in $\Gamma$, for every viable vertex $u$ of color $a$ or $a_{w'}$ we have that $(s,u)$ and $(u,s)$ are edges in $\tc(D)$ for a vertex $s\in C\cap S$. Hence, for every viable vertex $w$ of color $a_{w'}$ we have the edge $(v,w)$ in $\tc(D)$. There is at least one viable vertex, so we can pick that one.
	\item If $v',w'\in R$, we have $(v',w')=(v,w)\in E(\tc(D))$ as we would have discarded the template otherwise.
	\item If $v'\in R, w'\notin R$, then $v=v'$ and for every viable vertex $w$ of color $a_{w'}$ we have $(v,w)\in E(\tc(D))$ by the definition of viability. There is at least one viable vertex of color $a_{w'}$, so we can pick that one.
	\item Last, if $v'$ and $w'$ are not in the same SCC and $v'\notin R$, the viability of $v$ ensures that there is a suitable vertex $w$.
	\end{itemize}
	After this procedure, the outdegree of each vertex in $H$ is at most the outdegree of the vertex with the same color in $\Gamma$ and thus at most $T$. 
	Further, for all pairs of colors $a,b$ such that there is an edge from color $a$ to color $b$ in $\Gamma$ we have that \emph{all} vertices of color $a$ in $H$ have an edge to a vertex with color $b$ in $H$.
	Hence, for every commodity $(s,t)\in K$ there is an $s$-$t$-path in $H$ by starting at a vertex with the color of $s$ and iteratively taking an edge to a vertex with the next color on the path that satisfies $(s,t)$ in $\Gamma$ until we arrive at~$t$.
	Thus, $H$ witnesses $(D,K,T,c)$ and $(D,K,T)$ to be YES-instances of the corresponding problems.
	
	For the other direction assume there is a solution graph $H$ to $(D,K,T)$.
	By \cref{lem:commodities_bounded_sol} we can assume $H$ to consist of at most $2^k+k$ vertices. 
	By the above arguments we can further assume every SCC of size 2 or larger in $H$ to contain at least one vertex in $S$ (note that this does not increase the number of vertices and $H$ still consists of one $s$-$t$-path per commodity $(s,t)\in K$).
	We will try at least one colorful problem instance $(D,K,T,c)$ where $c$ assigns distinct colors to all vertices in $H$.
	Consider the template $\Gamma = D[V(H)]$ and note that it is not discarded.
	Further, note that each vertex in $H$ is marked viable w.r.t. $\Gamma$. Hence, the algorithm correctly decides that $(D,K,T$) is a YES-instance. 
	See \cref{fig:color_coding_rspp_by_comm} for an overview on the algorithm.
\end{proof}

\section{Structural Parameters}
\label{sec:structural}
In our final section, we turn towards a graph-structural analysis of the considered parcel sortation problems. At first glance, it may seem difficult to identify a ``reasonable'' graph-structural parameter that could be used to solve  \pfixed\ or \pvar: the \NP-hardness of both problems on an orientation of a star (\cite{VanDyk_2023} in combination with \cref{fact:pvar_pfix_pathlen_tw}) rules out not only the use of treewidth and its directed variants~\cite{GanianHKLOR14,GanianHK0ORS16}, but also 
 a range of other much more restrictive parameterizations such as treedepth~\cite{sparsity} and the vertex cover number~\cite{BhoreGMN20,CyganLPPS14,Zehavi17}.
 
The above situation is not unique: there is a well-known example of another routing problem---\textsc{Edge Disjoint Paths}---which suffers from a similar difficulty, notably by being para\NP-hard even on undirected graphs with vertex cover number $3$~\cite{FleszarMS18}. But for \textsc{Edge Disjoint Paths}, one can achieve fixed-parameter tractability when parameterizing by the combination of treewidth and the maximum degree~\cite{GanianOR21}, while for \pfixed\ and \pvar\ we can exclude tractability even under this combined parameterization by recalling~\cref{thm:pvar_pfix_deg_tw}. 

The above considerations raise the question: Are there any structural restrictions under which we can efficiently solve instances of \pfixed\ or \pvar\ involving a large number of commodities and target? As the main contribution of this section, we answer this question positively by identifying the \emph{maximum route length} as the missing ingredient required for tractability. Indeed, from an application perspective, restricting one's attention only to physical routes which do not involve too many intermediary nodes is well-aligned with practical considerations for the routing of most goods. For the purposes of this section we will assume that the parameterized instances of \pfixed\ and \pvar\ come equipped with an additional integer parameter \pathl called \emph{path length} and require:
\begin{itemize}
\item for \pfixed, the path $P$ of every routed commodity has length at most \pathl;
\item for \pvar, for each commodity $(s,t)\in K$, there is a directed $s$-$t$-path $P'$ in $H$ with a corresponding \emph{base path} $P$ in $D$ (formally, $P'\subseteq \tc(P)$), and $P$ has length at most $\pathl$. To formally distinguish it from \pvar, we call this enriched problem \pvarp (the problems coincide for $\pathl\geq |V(G)|$).
 \end{itemize}

As our first result, we establish the fixed-parameter tractability of both problems of 
interest when \pathl\ is included in the parameterization.
  
\begin{theorem}
\label{thm:deg_pathlen_tw}
	\pfixed and \pvarp are in \fpt w.r.t.\ the combination of the treewidth and maximum degree of $\undund{D}$ and the path length.
\end{theorem}

\begin{proof}
Let $(D, K, T)$ be an \pvarp or \pfixed instance with path length $p$.
Let $tw$ and $\delta_{\max}$ be the treewidth and maximum degree of $\undund{D}$, respectively. 

We first compute a width-optimal nice tree-decomposition $(\mathtt{T},\chi)$ using known techniques~\cite{BodlaenderK96,KorhonenL23}.
So, for the following consider a nice tree decomposition $(\mathtt{T},\chi)$ of $\undund{D}$ with treewidth at most $tw$.
We use the standard dynamic programming technique that starts at the leaves of the decomposition and traverses up to the root.
  For each node~$b$ of $\mathtt{T}$, let $N(b)$ denote the set of all vertices that are in $\chi(b)$ or have distance at most $p$ in $\undund{D}$ to a vertex in $\chi(b)$.
We define a \emph{state} with respect to a node $b$ as a pair $(d,R)$ of a function $d: N(b) \to \set{0, \dots, T}$ and a directed graph $R$ over the vertices in $N(b)$ with $\delta^+(R) \le T$.
We say that a state $(d,R)$ is a \emph{candidate} if there is a graph $H_d^R$ such that
\begin{enumerate}[(a)]
	\item  $H_d^R \subseteq \tc(D[c(b)\cup N(b)])$; 
	\item Every vertex has outdegree at most $T$ in $H_d^R$; 
	\item For \pvarp, for every commodity $(s,t) \in K$ such that $s\in c(b)$ or $t\in c(b)$, there is a directed $s$-$t$-path in $H_d^R$ with a base path of length at most $p$ in $D$ and for \pfixed, for every commodity $(s,t,P) \in K$ such that $s\in c(b)$ or $t\in c(b)$, there is a directed $s$-$t$-path in $H_d^R\cap \tc(P)$;
	\item $R = H_d^R[N(b)]$;
	\item Every vertex~$v \in N(b)$ has outdegree $d(v)$ in $H_d^R$.
\end{enumerate}

Clearly, if the root $r$ has a candidate $(d,R)$, then $H_d^R$ is a solution graph for $(D,K,T)$.
Further, if the instance $(D,K,T)$ is solved by some $H\subseteq \tc(D)$, the root has a candidate $(d,R)$, where $d$ assigns each vertex in $N(r)$ its outdegree in $H$ and $R = H[N(r)]$.
Thus, identifying whether the root has a candidate suffices to decide the instance.
We use the following dynamic program to identify all candidates at each node of $\mathtt{T}$.

At every node, we create a table that indicates whether a state $(d,R)$ of the node is a candidate, defined recursively as follows.
\begin{itemize}
	\item For a leaf, the only candidate is $(d,R)$, where $d$ is an empty function and $R$ the graph with $V(R)=\emptyset$. 
	 	 	\item For an \textsf{introduce} node~$\thickbar{b}$ that introduces a vertex~$v$ into a bag $\chi(b)$, a state $(\thickbar{d}, \thickbar{R})$ of $\thickbar{b}$ is a candidate if there exists a candidate $(d,R)$ of $b$, such that 
	\begin{itemize}
		\item $\thickbar{R} \subseteq \tc(D[N(\thickbar{b})])$; 
		\item $E(\thickbar{R}) \cap (N(b)\times N(b)) = E(R)$;
		\item For every $u \in N(b)$, we have $\thickbar{d}(u) = d(u) + |\set{(u,w)\in E(\thickbar{R}) \mid w\in  N(\thickbar{b}) \setminus  N(b)}| \le T$;
		\item For every $u\in N(\thickbar{b})\setminus N(b)$, we have $\thickbar{d}(u) \le T$ and $\thickbar{d}(u) = \delta^+_{\thickbar{R}}(u)$; 
		\item For \pvarp, for every $(s,v)\in K$ (or $(v,s)\in K)$ there is a directed $s$-$v$-path ($v$-$s$-path) in $\thickbar{R}$ with a base path in $D$ of length at most $p$, and for \pfixed, for every $(s,v,P)\in K$ (or $(v,s,P)\in K)$ there is a directed $s$-$v$-path ($v$-$s$-path) in $\thickbar{R}\cap \tc(P)$.
	\end{itemize}
	 	 	\item For a \textsf{forget} node~$\thickbar{b}$ that removes a vertex~$v$ from a bag $\chi(b)$, a state $(\thickbar{d}, \thickbar{R})$ of $\thickbar{b}$ is a candidate, if there exists a candidate $(d,R)$ of $b$, such that $\thickbar{R} = R[N(\thickbar{b})]$ and for every $u \in N(\thickbar{b})$, $\thickbar{d}(u) = d(u)$.
	 	 	\item For a \textsf{join} node~$\thickbar{b}$ that joins two nodes~$b$ and~$b'$, a state~$(\thickbar{d}, R)$ of~$\thickbar{b}$ is a candidate, if there exist candidates~$(d, R)$ and~$(d', R)$ of $b$ and $b'$, respectively, such that for all
	 $v \in N(\thickbar{b})$, $\thickbar{d}(v) = d(v) + d'(v) - \delta_R^+(v)$ and $\thickbar{d}(v) \le T$.
	\end{itemize}

We assume $\delta_{\max} > 1$, because otherwise $D$ is a union of disjoint edges, so the problem can easily be solved in linear time.
Then for each vertex~$v$ of $D$, $N(v)$ is loosely bounded by $\delta_{\max}^{p + 1}$.
Consider any node $b$.
We have $|\chi(b)| \le tw +1$ and $|N(b)| \le (tw+1)\delta_{\max}^{p+1}$.
For every state of $b$, we have that $d$ takes one of $T+1$ possible values for every $v\in N(b)$ and there are at most $(\sum_{i=0}^{T+1} \binom{N(b)}{i})^{N(b)}$ possible graphs for $R$. 
Therefore, the number of states of $b$ is less than
\[
	(T+1)^{|N(b)|} \cdot \left(\sum_{i=0}^{T+1}\binom{N(b)}{i}\right)^{N(b)}
	< (\delta_{\max}+1)^{(tw+1)\delta_{\max}^{p + 1}} \cdot 
	\left(\sum_{i=0}^{\delta_{\max}+1}\binom{(tw+1)\delta_{\max}^{p + 1}}{i}\right)^{(tw+1)\delta_{\max}^{p + 1}}
\]
where by \cref{obs:bounded_target}, we assume that $T < \delta_{\max}$.
Hence, the size of each table is bounded by some computable function $f(tw, \delta_{\max}, p)$.
Further, we note that by the above rules computing all candidates of each node in a dynamic programming manner from the leaves to the root takes \fpt time.

We conclude the proof by noting that the above algorithm correctly identifies the candidates at each node as shown in the claim below.

\begin{claim}\label{lem:candidates_correct}
	The algorithm correctly identifies all candidates for each node.
\end{claim}
 	We prove by induction along the tree decomposition.
	For some directed graph $H$ and a set of vertices $S\subseteq V(H)$ we define the graph $H[S]_{\tc} := \big(S,  \, E(H) \cap \tc(D[S])\big)$; that is, it is the subgraph of $H$ on $S$ that only keeps edges in the transitive closure of $D[S]$. 

	\paragraph*{Leaves.} Here, the only candidate is the empty function and the graph with no vertices.

	\paragraph*{\textsf{Introduce} Nodes.} Suppose node $\thickbar{b}$ introduces a vertex~$v$ into a bag $X(b)$.
	Suppose the algorithm sets $(\thickbar{d}, \thickbar{R})$ to be a candidate of $\thickbar{b}$ as witnessed by some $(d, R)$.
	By the induction hypothesis, $(d,R)$ is a candidate for $b$ witnessed by some graph $H_d^R$.
	We create a graph $H_{\thickbar{d}}^{\thickbar{R}}$ by combining $H_d^R$ and $\thickbar{R}$. In particular, $V(H_{\thickbar{d}}^{\thickbar{R}}) = c(\thickbar{b}) \cup N(\thickbar{b}) \supseteq V(H_d^R)$ and $E(H_{\thickbar{d}}^{\thickbar{R}}) = E(H_d^R) \cup E(\thickbar{R})$.
	Then (a) and (d) hold by construction and (b),(c), and (e) hold by combining the properties of $H_d^R$ with the five properties of $\thickbar{d}$ and $\thickbar{R}$ as defined in the algorithm's \textsf{introduce} step.
	
	For the other direction, let node $\thickbar{b}$ have a candidate $(\thickbar{d}, \thickbar{R})$ witnessed by some $H_{\thickbar{d}}^{\thickbar{R}}$. 
	Let $H_d^R = H_{\thickbar{d}}^{\thickbar{R}}[c(b)\cup N(b)]_{\tc}$, $d(v) = \delta^+_{H_d^R}$ for every $v\in N(b)$, and $R = H_d^R[N(b)]$.
	Then $(d,R)$ is a candidate for node $b$ as witnessed by $H_d^R$, where (a), (d), and (e) hold by construction, while (b) and (c) hold due to the candidate properties of $H_{\thickbar{d}}^{\thickbar{R}}$.
	In particular, for (c), we remark that all paths starting or ending in $X(b)$ of length at most $p$ only use vertices in $N(b)$ and thereby do not use any of the removed edges.
	By induction, the algorithm identifies $(d,R)$ as a candidate for node $b$.
	What remains to be shown is that the algorithm identifies $(d,R)$ as a witness for $(\thickbar{d},\thickbar{R})$ being a candidate of $\thickbar{b}$.
	To see this, first note that $\thickbar{R} \subseteq \tc(D[N(\thickbar{b})])$ as $(\thickbar{d},\thickbar{R})$ is a candidate. 
	By construction, $E(\thickbar{R}) \cap (N(b)\times N(b)) = R$ and, as $v\in c(\thickbar{R})$, condition (c) for $(\thickbar{d},\thickbar{R})$ implies that all commodities starting or ending in $v$ are satisfied by $\thickbar{R}$.
	Further, (b) gives that $\thickbar{d}(u) \le T$ for all $u\in N(\thickbar{b})$ and (e) gives for all $u\in N(\thickbar{b})\setminus N(b)$ that $\thickbar{d}(u) = \delta^+_{\thickbar{R}}(u)$, as all neighbors of $u$ are in $\thickbar{R}$.
	Last, for vertices in $N(b)$, by construction we have that $\thickbar{d}$ equals $d$ with the additional outdegree with edges to or from $N(\thickbar{b})\setminus N(b)$.
  	
	\paragraph*{\textsf{Forget} Nodes.} Suppose a node $\thickbar{b}$ removes a vertex~$v$ from a bag $X(b)$.
	Suppose the algorithm sets $(\thickbar{d}, \thickbar{R})$ to be a candidate of $\thickbar{b}$ as witnessed by some $(d, R)$.
	By induction, $(d,R)$ is a candidate for $b$ with a respective graph $H_d^R$.
	Let $H_{\thickbar{d}}^{\thickbar{R}} = H_d^R$.
	Then (b)--(e) hold by construction and (a) follows by $N(b) \setminus N(\thickbar{b}) \subseteq c(b)$, which holds as follows.
	Let $u\in N(b) \setminus N(\thickbar{b})$. 
	Then there is a path of length at most $p$ from $u$ to $v$ that does not visit any vertices in $X(\thickbar{b})$. Let $(u',v)$ be the last edge of this path. Thus, by the tree decomposition, $u'$ and $v$ have to share some bag and as $v\notin X(\thickbar{b})$, we have $u'\in c(b)$. By the same reasoning, as the path from $u$ to $u'$ does not pass through $X(b)$, we have $u\in c(b)$.
	
	For the other direction, let node $\thickbar{b}$ have a candidate $(\thickbar{d}, \thickbar{R})$ witnessed by some $H_{\thickbar{d}}^{\thickbar{R}}$. 
	Let $H_d^R = H_{\thickbar{d}}^{\thickbar{R}}$, for every $v\in N(b)$ let $d(v) = \delta_{H_d^R}^+(v)$, and $R = H_d^R[N(b)]$.
	Then $(d,R)$ is a candidate for node $b$ as witnessed by $H_d^R$, where (a) holds as $c(b)=c(\thickbar{b})$ and $N(\thickbar{b}) \subseteq N(b)$, and $(b)$--$(e)$ hold by construction.
	By induction, the algorithm identifies $(d,R)$ as a candidate for node $b$.
	Further, as $\thickbar{R} = R[N(\thickbar{b})]$, the algorithm identifies $(\thickbar{d},\thickbar{R})$ as a candidate for node $\thickbar{b}$ as witnessed by $(d,R)$.

	\paragraph*{\textsf{Join} Nodes.}  Suppose a node $\thickbar{b}$ joins two nodes $b$ and $b'$.
	Suppose the algorithm sets $(\thickbar{d}, R)$ to be a candidate of $\thickbar{b}$ as witnessed by some $(d, R)$ and $(d',R)$ that are, by induction, candidates of $b$ and $b'$, respectively.
	Let $H_d^R$ and $H_{d'}^R$ be the graphs witnessing them as candidates, respectively.
	Define $H_{\thickbar{d}}^R$ as their union graph.
	Conditions (a), (c), and (d) follow immediately from the candidate properties of $H_d^R$ and $H_{d'}^R$.
	For condition (b), let $v\in c(\thickbar{b}) \cup N(\thickbar{b})$.
	Note that $H_{d}^R, H_{d'}^R$, and $H_{\thickbar{d}}^R$ share the induced subgraph $R$.
	As $H_d^R$ and $H_{d'}^R$ fulfill (b), $\delta_{H_{\thickbar{d}}^R}^+(v) > T$ would imply that  $\Delta_{H_{\thickbar{d}}^R}^+(v) \setminus N(\thickbar{b}) \cap c(b) \neq \emptyset$ and $\Delta_{H_{\thickbar{d}}^R}^+(v) \setminus N(\thickbar{b}) \cap c(b') \neq \emptyset$. 
	In particular, $v$ would have edges to both $c(b)\setminus X(\thickbar{b})$ and $c(b')\setminus X(\thickbar{b})$. 
	By the tree decomposition this implies $v\in X(\thickbar{b})$, which in turn gives that all its neighbors are in $N(\thickbar{b})$.
	Thus there can be no such $v$ and (b) holds.
	For condition (e), let $v\in N(\thickbar{b})$. 
	Note that
	$\Delta_{H_{\thickbar{d}}^{\thickbar{R}}}(v) \subseteq  V(H_d^R) \cup V(H_{d'}^R)$, 
	$|\Delta_{H_{\thickbar{d}}^{\thickbar{R}}}(v) \cap V(H_d^R)| = d(v)$, and
	$|\Delta_{H_{\thickbar{d}}^{\thickbar{R}}}(v) \cap V(H_{d'}^R)| = d'(v)$.
	Further, 
	$\Delta_{H_{\thickbar{d}}^{\thickbar{R}}}(v) \cap V(H_d^R) \cap V(H_{d'}^R) = \Delta_R^+(v)$, as $u\in c(b)$ and $u\in c(b')$ imply $u\in X(\thickbar{b})$ by the tree decomposition.
	Hence, $|\Delta_{H_{\thickbar{d}}^{\thickbar{R}}}(v)| = d(v)+d'(v)- \delta_R^+(v)$ and the algorithm correctly identifies $(\thickbar{d}, R)$ as a candidate.
	
	For the other direction, let node $\thickbar{b}$ have a candidate $(\thickbar{d}, R)$ witnessed by some $H_{\thickbar{d}}^R$. 
	Let $H_d^R = H_{\thickbar{d}}^R[c(b)\cup N(\thickbar{b})]_{\tc}$ and $H_{d'}^R = H_{\thickbar{d}}^R[c(b')\cup N(\thickbar{b})]_{\tc}$.
	Let $d$ and $d'$ assign every vertex in $N(\thickbar{b})$ its outdegree in $H_d^R$ or $H_{d'}^R$, respectively.
	Then $(d,R)$ and $(d',R)$ are candidates for nodes $b$ and $b'$, respectively, as witnessed by $H_d^R$ and $H_{d'}^R$, where (a), (b), (d), and (e) hold by construction and (c) holds as every path of length at most $p$ that uses a vertex in $c(b)$ (or $c(b')$) cannot use a vertex outside $c(b)\cup N(\thickbar{b})$ (or $c(b)\cup N(\thickbar{b})$) and no edge outside $\tc(D[c(b)\cup N(\thickbar{b})])$ (or $\tc(D[c(b)\cup N(\thickbar{b})])$).
	By induction, the algorithm correctly identifies $(d,R)$ and $(d',R)$ as candidates.
	By the same reasoning as above we have $\thickbar{d}(v) = |E_v| = d(v)+d'(v)- |\set{u\mid (v,u)\in E(R)}|$ for all $v\in N(\thickbar{b})$, so the algorithm correctly identifies $(\thickbar{d},R)$ as a candidate for node $\thickbar{b}$ with $(d,R)$ and $(d',R)$ as witnesses.
	   \end{proof}

Note that the constructions above can be used for a backtracking approach that allows us to reconstruct the graph $H_d^R$ for any candidate $(d,R)$.
Thus, the dynamic program described in the proof of \cref{thm:deg_pathlen_tw} can not only solve the decision problem but also output a solution graph~$H$.

We complement \cref{thm:deg_pathlen_tw} by showing that not bounding any one of the three parameters (treewidth, degree, and path length) yields \nphardness for both problems.
\cref{thm:pvar_pfix_deg_tw} in \cref{sec:target} already establishes this for instances of large path length. 
Van Dyk, Klause, Koenemann, and Megow~\cite{VanDyk_2023} proved \pfixed to be \nphard\ even when the input graph is an orientation of a (high-degree) star.
As in such graphs, each commodity has a 
  unique path connecting its source and destination, the same construction can be used for \pvarp.
\begin{fact} 	\label{fact:pvar_pfix_pathlen_tw}
		\pvarp and \pfixed with path length $p=2$ on stars are \nphard.
\end{fact}
 Last, we prove \NP-hardness for instances with constant degree and path length by
 reducing from \threeSATtwo: in this \np-complete restriction of \threeSAT, each clause has exactly three literals and each literal occurs in exactly two clauses~\cite{Darmann_2021}.

\begin{theorem}
\label{thm:pvar_deg_pathlen}
	\pvarp is \nphard, even when restricted to instances with $T \leq 2$, path length at most $4$, and maximum degree in $\undund{D}$ at most $7$.
\end{theorem}

	\begin{figure}[t]
		\centering
	\includegraphics[page=1, scale=.84]{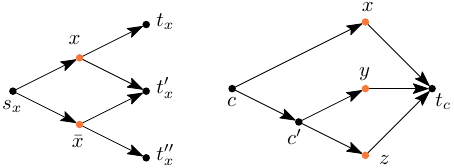}
	\caption{Variable gadget (left) and clause gadget (right) for the proof of~\cref{thm:pvar_deg_pathlen}. Orange colored vertices are shared between different gadgets.} 	\label{fig:pvar_deg_pathlen}
	\end{figure}

\begin{proof}
	Suppose we have an instance of \threeSATtwo with a set $C$ of $m$ clauses and a set $X$ of $n$ variables.
	Construct $D$ as follows.
	For every variable~$x\in X$, we use the variable gadget as in~\cref{fig:pvar_deg_pathlen} (left) and add the following commodities to $K$: $(s_x, t)$ for $t \in \set{t_x, t'_x, t''_x}$.
	For every clause~$c\in C$ over the literals $x, y$, and $z$, we use the clause gadget as in~\cref{fig:pvar_deg_pathlen} (right). There, the vertices of $x, y$, and $z$ correspond to the respective literal vertex ($x$ or $\thickbar{x}$) in the variable gadgets. 
	Further, we add the following commodities to $K$: $(c, t)$ for $t \in \set{x, y, z, t_c}$.
	Observe that literal vertices have the highest degree (3 in- and 4 out-edges) as they are part of one variable and two clause gadgets. 
	The constructed \pvarp instance with $\pathl = 4$ has a solution with target at most $2$ if and only if the \threeSATtwo instance is a YES-instance.

	Suppose we have a YES-instance of \threeSATtwo. 
	Consider the graph $H$ such that $V(H) = V(D)$ and
	\begin{align*}
		E(H) &= \bigcup_{x\in X} E_x \cup \bigcup_{c\in C} E_c, \\
		E_x &= \begin{cases}
			\set{(s_x, \thickbar{x}), (s_x, t_x), (\thickbar{x},t_x'), (\thickbar{x},t_x'')}, \quad \text{if $x$ is \True;}\\
			\set{(s_x, x), (s_x, t_x''), (x,t_x), (x,t_x')}, \quad \text{else;}\\
			\end{cases}\\ 
		E_c &= \set{(c,c'), (c,x), (c',y), (c',z), (a,t_c)}, 
	\end{align*} 
	where $a\in \set{x,y,z}$ is any \True literal in the clause. 
	Observe that $H$ is a subgraph of $\tc(D)$ and for every commodity $(s,t)\in K$ there is a directed $s$-$t$-path in $H$.
	Further, the outdegree of vertices $s_x, c$, $c'$ for a variable $x$ or clause $c$ are 2 and the outdegrees of $t_x, t_x', t_x'', t_c$ are 0. Each literal vertex ($x$ or $\thickbar{x}$) that is \False has 2 outdegree in its variable gadget and no outdegree towards a clause gadget vertex. 
	Each literal vertex that is \True has no outdegree inside its variable gadget and at most outdegree 1 in each of the two clause gadgets it appears in.
	Thus, $H$ is a solution for target 2.

	Now suppose there is a solution graph $H$ for the \pvar instance with target 2.
	Consider the gadget of a variable $x$. 
	To fulfill all commodities inside the gadget while preserving a maximum outdegree of 2 in $H$, we have $\set{(x,t_x), (x,t_x')} \subseteq E(H)$ or $\set{(\thickbar{x},t_x'), (\thickbar{x},t_x'')}\subseteq E(H)$. We let $x$ be \False if the first holds and $x$ be \True if only the second holds.
	This variable assignment from all variable gadget satisfies all clauses as follows.
	Consider any clause $c$ over the literals $x,y,z$. 
	To fulfill all commodities inside the gadget while preserving a maximum outdegree of 2 in $H$,
	we have $\set{(c,x), (c,c'), (c',y), (c',z)} \subseteq E(H)$.
	Then, $c$ and $c'$ cannot have an edge to $t_c$ in $E(H)$ without exceeding the target of 2.
	Thus, $(a,t_c)\in E(H)$ for at least one $a\in \set{x,y,z}$. 
	This literal has to be \True (and thereby satisfies the clause) because otherwise literal $a$ would have at least outdegree 3 in $H$ (two in its variable gadget and one in the clause gadget).
\end{proof} 

We modify the proof of~\cref{thm:pvar_deg_pathlen} using more sophisticated gadgets to account for the fixed paths and obtain a similar result for \pfixed.

\begin{theorem}
	\label{thm:pfixed_deg_pathlen}
	\pfixed is \nphard, even when restricted to instances with $T \leq 4$, path length at most $5$ and maximum degree in $\undund{D}$ at most $11$.
\end{theorem}

\begin{proof}
	Suppose we have an instance of \threeSATtwo with a set $C$ of $m$ clauses and a set $X$ of $n$ variables.
	Construct $D$ as follows.
	For every variable~$x$, we use the variable gadget as in~\cref{fig:pfixed_deg_pathlen} (left) and add the following commodities to $K$: $(s_x, t)$ for $t \in \set{t^x_1, t^x_2, t^x_3, t^{\thickbar{x}}_1, t^{\thickbar{x}}_2, t^{\thickbar{x}}_3, t^{\thickbar{x}}_4}$.
	The directed paths for these commodities are exactly the unique directed paths in the gadget.
	For every clause~$c$ over the literals $x,y$, and $z$, we use the clause gadget as in~\cref{fig:pfixed_deg_pathlen} (right) and add the following commodities to~$K$. 
	First, add $(a, t)$ for $a \in \set{c, c', c'', \tilde{c}}$ and $t \in \set{t^a_1, t^a_2}$.
	Second, add $(c,c')$ and $(c',c'')$.
	All of these commodities are assigned their unique respective path.
	Last, add $(c, t)$ for $t \in \set{t^{c,a}_1, t^{c,a}_2 | a \in \set{x, y, z}}$. 
	For $a\in \set{x,y,z}$, the directed path for $(c, t_1^{c,a})$ (or $(c, t_2^{c,a})$) is the unique directed path from $c$ to $t_1^{c,a}$ (or $t_2^{c,a}$) that contains $a$.
	Observe that vertices $\tilde{c}$ have the highest outdegree (3 in- and 8 out-edges).

	\begin{figure}[bt]  	\centering
	\includegraphics[page=2, scale=1]{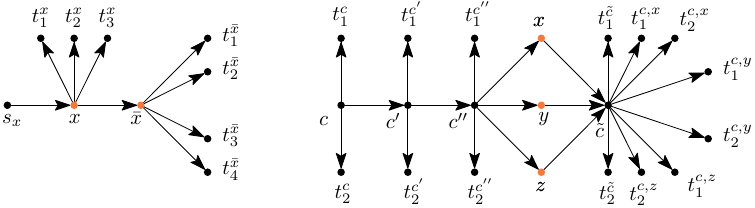}
	\caption{Variable gadget (left) and clause gadget (right) for the proof of~\cref{thm:pfixed_deg_pathlen}.
	Orange colored vertices are shared between different gadgets.}
	\label{fig:pfixed_deg_pathlen}
	\end{figure}

	Suppose we have a YES-instance of \threeSATtwo. 
	Consider the graph $H$ such that 
	\begin{align*}
		E(H) &= \bigcup_{x\in X} E_x \cup \bigcup_{c\in C} E_c,\\
		E_x &= \begin{cases}
			\set{(s_x, \thickbar{x})} \cup \set{(s_x, t_i^x) \mid t\in [3]} \cup \set{(\thickbar{x}, t_i^{\thickbar{x}}) \mid t\in [4]},
				 \hphantom{,(x,t_4)}\quad \text{if $x$ is \True;}\\
			\set{(s_x, x), (x,t_4^{\thickbar{x}})} \cup \set{(x, t_i^x) \mid t\in [3]} \cup \set{(s_x, t_i^{\thickbar{x}}) \mid t\in [3]}, 
				\quad \text{else.}\\
			\end{cases}\\ 
		E_c &= \set{(a, t_1^a), (a, t_2^a) \mid a \in \set{c, c', c'', \tilde{c}}} 
		  \cup \set{(c,c'),(c',c''),(c,\tilde{c}),(c',x), (x,t_1^{c,x}), (x,t_2^{c,x})}\\
		 & \hphantom{= } \cup \set{(c'', t_1^{c,y}), (c'', t_2^{c,y}), (\tilde{c},t_1^{c,z}), (\tilde{c}, t_2^{c,z})}, 
	\end{align*} 
	where without loss of generality $x$ is a \True literal in the clause.
	Observe that $H$ is a subgraph of $\tc(D)$ and for every commodity $(s,t, P)\in K$ there is a directed $s$-$t$-path in $H\cap \tc(P)$.
	Further, the outdegree of vertices $t_i^x$ is 0 and the outdegree of vertices $s_x, c, c', c'',$ and $\tilde{c}$ is 4. 
	Each literal vertex ($x$ or $\thickbar{x}$) that is \False has 4 outdegree in $E_x$ and no outdegree in any $E_c$.
	Each literal vertex that is \True has no outdegree in $E_x$ and at most outdegree 2 in each of $E_{c}$, $E_{c'}$, the two clause gadgets it appears in.
	Thus, $H$ is a solution for target 4.

	Now suppose there is a solution graph $H$ for the \pfixed instance with target 4.
	Consider the gadget of a variable $x$ and partition the vertices in two sets $S_x = (s_x, x, \thickbar{x})$ and $T_x = \set{t_1^x,t_2^x,t_3^x,t_1^{\thickbar{x}},t_2^{\thickbar{x}},t_3^{\thickbar{x}},t_4^{\thickbar{x}}}$.
	We let $x$ be \True if and only if there are at most 2 edges $(x,t)\in H$ such that $t\in T_x \cup \{\thickbar{x}\}$.
	Note that in this case there are more than 2 edges $(\thickbar{x},t)\in H$ such that $t\in T_x$ as otherwise not all commodities in the gadget would be satisfied.
	This variable assignment from all variable gadget satisfies all clauses as follows.
	Consider any clause $c$ over the literals $x,y,z$. 
	To fulfill all commodities inside the gadget,
	we have $\set{(a, t_1^a), (a, t_2^a) \mid a \in \set{c, c', c'', \tilde{c}}}\subseteq E(H)$ as well as $(c,c'),(c',c'')\in E(H)$.
	This already gives $c$ and $c'$ an outdegree of 3 and $c''$ an outdegree of 2.
	Thus, there can be at most 4 more edges starting in $S_x$ to satisfy the 6 remaining commodities $(c,t), t\in T_x$.
	We can use one of the 4 edges to connect to $\tilde{c}$, which still has capacity for 2 not yet accounted edges. 
	Still one of the 6 remaining commodities is unsatisfied, for which the only remaining option is to use the vertices $x,y,$ or $z$.
	Suppose all literals in the clause were \False. 
	Then, in particular, $x,y,z$ have already used at least 3 of their 4 out-edges towards vertices in their respective variable gadget.
	As connecting $c$ to any of $x,y,$ or $z$ takes one edge each and only enables at most one edge each, it is not possible to satisfy all 6 remaining commodities.
	Therefore, at least one of $x,y,$ or $z$ uses at most 2 edges in its variable gadget, and therefore, it is \True and satisfies the clause.
\end{proof}

\section{Concluding Remarks and Discussion}
\label{sec:conc}
The complexity-theoretic behavior of both \pvar and \pfixed w.r.t.\ graph-structural parameterizations of is somewhat unique: there seem to be no structural parameterizations of $D$ alone which would yield \XP\ or fixed-parameter algorithms without bounding the input size. 
 Indeed, we achieve tractability in \cref{thm:deg_pathlen_tw} only by the addition of the path length as a parameter. 
Luckily, this parameter is well-motivated, as in practice virtually all commodities are routed via rather short paths in the fulfillment network due to its hierarchical nature. 
 
The main technical open question arising from our work is the classical complexity of \pfixed\ in the baseline case of $T=1$. Another aspect which we consider noteworthy is that \pfixed\ and \pvar\ seem to have the same parameterized complexity landscape: they are both \FPT\ w.r.t.\ $|K|$, para\NP-hard w.r.t.\ $T$, and appear indistinguishable from the perspective of structural parameterizations (see Table~\ref{table:results}). Is there a deeper connection between the two problems and could this be exploited to obtain unified algorithms for both problems which avoid, e.g., the heavy machinery of Ramsey's
 Theorem and color coding?  
 
\bibliography{refs} 

\end{document}